%% file: main.tex
\documentclass[a4paper,UKenglish, autoref, thm-restate]{lipics-v2021}

\input{packages}

\input{macros}

\newif \iffull
\fulltrue

\newcommand{\ERCagreement}{This paper is part of a project that have received funding from the European Research Council (ERC) (grant agreement No 948057 -- {\sc BOBR}).}

\hideLIPIcs  

\title{Canonical decompositions in monadically stable and bounded shrubdepth graph classes}

\titlerunning{Canonical decompositions of stable graphs} 

\author{Pierre Ohlmann}{Institute of Informatics, University of Warsaw, Poland}{}{}{}
\author{Micha{\l} Pilipczuk}{Institute of Informatics, University of Warsaw, Poland}{}{}{}
\author{Wojciech Przybyszewski}{Institute of Informatics, University of Warsaw, Poland}{}{}{}
\author{Szymon Toru{\'n}czyk}{Institute of Informatics, University of Warsaw, Poland}{}{}{}

\authorrunning{P. Ohlmann, M. Pilipczuk, W. Przybyszewski and S. Toru{\'n}czyk} 

\Copyright{Jane Open Access and Joan R. Public} 



\category{} 

\relatedversion{} 

\funding{\ERCagreement}

\acknowledgements{We are indebted to Pierre Simon for enlightening discussions about classic results in stability theory that greatly helped in deriving the results presented in this work. We also thank Jakub Gajarsk\'y, Rose McCarty, and Marek Sokołowski for their contribution in discussions around the topic of this work.}

\nolinenumbers 

\EventEditors{John Q. Open and Joan R. Access}
\EventNoEds{2}
\EventLongTitle{42nd Conference on Very Important Topics (CVIT 2016)}
\EventShortTitle{CVIT 2016}
\EventAcronym{CVIT}
\EventYear{2016}
\EventDate{December 24--27, 2016}
\EventLocation{Little Whinging, United Kingdom}
\EventLogo{}
\SeriesVolume{42}
\ArticleNo{23}

\begin{document}

\maketitle

\begin{abstract}
    We use model-theoretic tools originating from stability theory to derive a result we call the Finitary Substitute Lemma, which intuitively says the following. Suppose we work in a stable graph class $\Cc$, and using a first-order formula $\varphi$ with parameters we are able to define, in every graph $G\in \Cc$, a relation $R$ that satisfies some hereditary first-order assertion $\psi$. Then we are able to find a first-order formula $\varphi'$ that has the same property, but additionally is {\em{finitary}}: there is finite bound $k\in \N$ such that in every graph $G\in \Cc$, different choices of parameters give only at most $k$ different relations $R$ that can be defined using $\varphi'$.
	
	We use the Finitary Substitute Lemma to derive two corollaries about the existence of certain canonical decompositions in classes of well-structured graphs.
	\begin{itemize}
	 \item We prove that in the Splitter game, which characterizes nowhere dense graph classes, and in the Flipper game, which characterizes monadically stable graph classes, there is a winning strategy for Splitter, respectively Flipper, that can be defined in first-order logic from the game history. Thus, the strategy is canonical.
	 \item We show that for any fixed graph class $\Cc$ of bounded shrubdepth, there is an $\Oh(n^2)$-time algorithm that given an $n$-vertex graph $G\in \Cc$, computes in an isomorphism-invariant way a structure $H$ of bounded treedepth in which $G$ can be interpreted. A corollary of this result is an $\Oh(n^2)$-time isomorphism test and canonization algorithm for any fixed class of bounded~shrubdepth. 
	\end{itemize}
\end{abstract}


\input{intro}

\input{preliminaries}
\input{splitter}
\input{stability}
\input{canonisation}
\iffull
\input{definable-strategies}
\else \fi

\bibliography{ref}

\iffull
\newpage
\appendix
\input{appendix}

\else \fi

\end{document}

%% file: packages.tex
\usepackage{amsmath, amsfonts,xfrac, amsthm,amssymb,fancyhdr}

\usepackage{mathtools}
\usepackage{setspace}

\usepackage{xypic,colortbl}

\usepackage{fixltx2e, ifthen, sidecap, wrapfig, xspace}

\usepackage{mathrsfs, stmaryrd}

\usepackage{
 calc, enumerate, makeidx,bbold
}

\usepackage[font=small,format=plain,labelfont=sc,textfont=it]{caption}
\usepackage[T1]{fontenc}
\usepackage[utf8]{inputenc}
\usepackage[cmyk]{xcolor}

\usepackage{todonotes}
\usepackage{cleveref}

%% file: macros.tex
\newcommand{\Oof}{\mathcal{O}}
\newcommand{\Oh}{\Oof}
\renewcommand{\preceq}{\preccurlyeq}
\newcommand{\Labels}{\mathsf{Labels}}
\newcommand{\Adj}{\mathsf{Adj}}
\newcommand{\lbl}{\mathsf{label}}
\newcommand{\Forb}{\mathsf{Forb}}

\newcommand{\cn}{\mathfrak{c}}

\newcommand{\refine}{\bar \cap}

\newcommand{\mathsym}[1]{{}}

{\begin{figure}[#1]
\centering 
\includegraphics[scale=#3]{#2}
}
{\end{figure}}

\newenvironment{romanlist'}[0]
{\begin{list}{\makebox[0.5cm][l]{\textit{\roman{enumi}')}}}{\usecounter{enumi}}}
{\end{list}}

\newcommand{\savelabel}[2]{\expandafter\newtoks\csname#1\endcsname
  \global\csname#1\endcsname={#2} \label{#1} #2}
\newcommand{\loadlabel}[1]{\noindent {\bf Lemma~\ref{#1}. } \textit{\the\csname#1\endcsname}
\medskip

}
\renewcommand{\setminus}{-}

\newcommand{\loadlabelthm}[1]{\medskip\noindent {\bf Theorem~\ref{#1}. }
  \noindent  \textit{\the\csname#1\endcsname}
\medskip
}

\newcommand{\loadlabelprop}[1]{\noindent {\bf Proposition~\ref{#1}. }
  \noindent  \textit{\the\csname#1\endcsname}
\medskip

}

\newcommand{\R}{\mathbb{R}}

\newcommand{\N}{\mathbb{N}}

\renewcommand{\subset}{\subseteq}

\newcommand{\atleast}[1]{{\ge n}}
\newcommand{\less}[1]{{<n}}

	\newcommand{\notacol}[2]{}

\newsavebox{\quoteitbox}

{\hspace*{\fill}{\upshape(\usebox{\quoteitbox})}\end{quote}%
\end{sloppypar}\noindent\ignorespacesafterend}

\newenvironment{quoteit*} 
{\begin{sloppypar}\noindent\slshape\begin{quote}\itshape} 
	{\end{quote}\ignorespaces\end{sloppypar}\noindent\ignorespacesafterend}

\newenvironment{quotetag*}
{~\par
	\begingroup                  
	\begin{equation*}
		 \begin{minipage}[c]{115mm}
			\it\noindent{\par}
}
{
		\end{minipage}
	\end{equation*}
	\endgroup                        
\par
\textnormal
\medskip
}

\newcommand{\Aa}{\ensuremath{\mathcal A}\xspace}

\newcommand{\Cc}{{\mathscr C}}

\newcommand{\CC}{\Cc}
\newcommand{\DD}{{\mathscr D}}
\newcommand{\Dd}{\ensuremath{\mathcal D}\xspace}

\newcommand{\Ff}{{\mathcal F}}

\newcommand{\Hh}{{\mathcal H}}

\newcommand{\Ll}{{\mathcal L}}

\newcommand{\Oo}{{\mathcal O}}
\newcommand{\Pp}{{\mathcal P}}
\newcommand{\Qq}{{\mathcal Q}}
\newcommand{\Rr}{{\mathcal R}}

\newcommand{\Ts}{{\mathsf T}}
\newcommand{\Is}{{\mathsf I}}

\newcommand\set[1]{\ensuremath{\{#1\}}}

\newcommand{\tn}[1]{\footnotesize{#1}}

\newtheoremstyle{theoremstyle}
  {3pt}
  {3pt}
  {\itshape}
  {0pt}
  {\bfseries}
  {}
  {4pt}
  {}

\numberwithin{equation}{section}

\newcommand{\enc}{\text{encodes}}

\newcommand{\wh}{\widehat}

\newlength{\wideaslength}

\renewcommand{\subset}{\subseteq}

\newcommand{\seta}[1]{}

\def\lsim{\mathrel{\rlap{\lower4pt\hbox{\hskip1pt$\sim$}}
    \raise1pt\hbox{$<$}}}                

\definecolor{gray1}{rgb}{0.99,0.99,0.99}
\definecolor{gray2}{rgb}{0.97,0.97,0.97}
\definecolor{gray3}{rgb}{0.95,0.95,0.95}
\definecolor{gray4}{rgb}{0.93,0.93,0.93}
\definecolor{gray5}{rgb}{0.91,0.91,0.91}
\definecolor{gray6}{rgb}{0.89,0.89,0.89}
\definecolor{gray7}{rgb}{0.87,0.87,0.87}
\definecolor{gray8}{rgb}{0.85,0.85,0.85}
\definecolor{gray9}{rgb}{0.83,0.83,0.83}
\definecolor{gray10}{rgb}{0.81,0.81,0.81}
\definecolor{gray20}{rgb}{0.71,0.71,0.71}
\definecolor{gray40}{rgb}{0.51,0.51,0.51}

\def\xInd#1#2{#1\setbox0=\hbox{$#1x$}\kern\wd0\hbox to 0pt{\hss$#1\mid$\hss}
	\lower.9\ht0\hbox to 0pt{\hss$#1\smile$\hss}\kern\wd0}

\def\xind{\mathop{\mathpalette\xInd{}}} 

\def\xnotind#1#2{#1\setbox0=\hbox{$#1x$}\kern\wd0
	\hbox to 0pt{\mathchardef\nn=12854\hss$#1\nn$\kern1.4\wd0\hss}
	\hbox to 0pt{\hss$#1\mid$\hss}\lower.9\ht0 \hbox to 0pt{\hss$#1\smile$\hss}\kern\wd0}

\def\xnind{\mathop{\mathpalette\xnotind{}}} 
\newcommand{\ind}[2][]{\xind_{#1}^{#2}}
\newcommand{\nind}[2][]{\xnind_{#1}^{#2}}

\newcommand{\tup}{\bar}

\newcommand{\from}{\colon}
\newcommand{\str}[1]{\mathbf{#1}}

\renewcommand{\le}{\leqslant}
\renewcommand{\ge}{\geqslant}
\renewcommand{\leq}{\leqslant}
\renewcommand{\geq}{\geqslant}
\renewcommand{\phi}{\varphi}

\newcommand{\tp}{\mathrm{tp}}

\DeclareMathOperator{\dist}{dist}

%% file: intro.tex
\section{Introduction}\label{sec:intro}


{\em{Stability theory}} is a well-established branch of model theory devoted to the study of stable theories, or equivalently classes of structures that are models of such theories. Here, we say that a formula\footnote{All formulas considered in this paper are first-order, unless explicitly stated.} $\varphi(\tup x;\tup y)$ is {\em{stable}} on a class of relational structures $\Cc$ if there is an integer $k\in \N$ such that for every $\str M\in \Cc$, one cannot find tuples $\tup u_1,\ldots,\tup u_k\in \str M^{\tup x}$ and $\tup v_1,\ldots,\tup v_k\in \str M^{\tup y}$ such that for all $i,j\in \{1,\ldots,k\}$,
\[
	\str M\models \varphi(\tup u_i,\tup v_j)\qquad \textrm{if and only if} \qquad i<j.
\]
Then $\Cc$ is stable if every formula is stable on $\Cc$. Intuitively, this means that using a fixed formula, one cannot interpret arbitrarily long total orders in structures from $\Cc$.
We refer to the textbooks of Pillay~\cite{pillay1996geometric} or of Tent and Ziegler~\cite{tent_ziegler} for an introduction to stability. 

The goal of this paper is to use certain classic results of stability theory, particularly the understanding of {\em{forking}} in stable theories, to derive statements about the existence of canonical decompositions in certain classes of well-structured graphs. Here, we model graphs as relational structures with one binary adjacency relation that is symmetric.

\subparagraph*{Finitary Substitute Lemma.} Our main model-theoretic tool is the Finitary Substitute Lemma, which we state below in a simplified form; see Lemma~\ref{lem:mt-core} for a full statement. 

To state the lemma, we need some definitions. A formula $\varphi(\tup x;\tup y)$ is {\em{finitary}} on a class of structures $\Cc$ if there exists $k\in \N$ such that for every $\str M\in \Cc$, we have
\[
	\left|\left\{\varphi(\str M^{\tup x},\tup v)\colon \tup v\in \str M^{\tup y}\right\}\right|\leq k,
\]
where $\varphi(\str M^{\tup x},\tup v)=\{\tup u\in \str M^{\tup x}~|~\str M\models \varphi(\tup u,\tup v)\}$. In other words, $\varphi(\tup x;\tup y)$ is finitary on $\Cc$ if by substituting different parameters for $\tup y$ in any model $\str M \in \Cc$, one can define only at most $k$ different relations on $\tup x$-tuples. Next, a sentence $\psi$ is {\em{hereditary}} if for every model $\str M$ and its induced substructure $\str M'$, $\str M\models \psi$ implies $\str M'\models \psi$. Finally, for a relation $R(\tup x)$ present in the signature, formula $\varphi(\tup x)$ (possibly with parameters), and sentence $\psi$, by $\psi[R(\tup x)/\varphi(\tup x)]$ we mean the sentence derived from $\psi$ by substituting every occurrence of $R$ with formula $\varphi$.

\begin{lemma}[Finitary Substitute Lemma, simplified version]
 Let $\Cc$ be a stable class of structures. Suppose $\varphi(\tup x;\tup y)$ is a formula and $\psi$ a hereditary sentence such that for every~$G\in \Cc$,
 \begin{equation}\label{eq:wydra}\textrm{there exists }\tup s\in G^{\tup y}\textrm{ such that }G\models \psi[R(\tup x)/\varphi(\tup x;\tup s)].
 \end{equation}
 Then there exists a formula $\varphi'(\tup x,\tup z)$ that also satisfies~\eqref{eq:wydra}, but is additionally finitary on $\Cc$.
\end{lemma}

Thus, intuitively, the Finitary Substitute Lemma says that in stable classes, every relation that is definable with parameters can be replaced by a finitary one, as long as we care that the relation satisfies some hereditary first-order assertion. The main observation of this paper is that this can be used in the context of various graph decompositions. Intuitively, if every step in decomposing the graph can be defined by a first-order formula with parameters, and the validity of the step can be verified using a hereditary first-order sentence, then we can use the Finitary Substitute Lemma to derive an equivalent definition of a step that is finitary. This yields only a bounded number of different steps that can be taken, making it possible to construct a decomposition that, in a certain sense, is canonical.

\subparagraph*{Classes of bounded shrubdepth.} Our first application concerns classes of {\em{bounded shrubdepth}}. The concept of shrubdepth was introduced by Ganian et al.~\cite{GanianHNOMR12} to capture dense graphs that are well-structured in a shallow way. On one hand, classes of bounded shrubdepth are exactly those that can be interpreted, using first-order formulas with two free variables, in classes of forests of bounded depth. On the other hand, graphs from any fixed class of bounded shrubdepth admit certain decompositions, called {\em{connection models}}, which are essentially clique expressions of bounded depth. (See Section~\ref{subsec:sd-prelim} for a definition of a connection model.) Thus, in particular every graph class of bounded shrubdepth has bounded cliquewidth, but classes of bounded shrubdepth are in addition stable~\cite{GanianHNOMR12}.

Shrubdepth is a dense counterpart of {\em{treedepth}}, defined as follows: the treedepth of a graph $G$ is the smallest integer $d$ such that $G$ is a subgraph of the ancestor/descendant closure of a rooted forest of depth at most $d$. In particular, every class of graphs of bounded treedepth has bounded shrubdepth; boundedness of treedepth and of shrubdepth is in fact equivalent assuming that the class excludes some biclique $K_{t,t}$ as a subgraph~\cite{GanianHNOMR12}. In essence, treedepth is a bounded-depth counterpart of treewidth in the same way as shrubdepth is a bounded-depth counterpart of cliquewidth.

In spite of the above, the combinatorial properties of shrubdepth are still much less understood than those of treedepth. For instance, a good understanding of subgraph obstacles allows one to construct suitable canonical decompositions for graphs of bounded treedepth. This allowed Bouland et al.~\cite{BoulandDK12} to design a graph isomorphism test that works in fixed-parameter time parameterized by the treedepth, or more precisely, in time $f(d)\cdot n^3\log n$, where $f$ is a computable function. While it is known that every class of bounded shrubdepth can be characterized by a finite number of forbidden induced subgraphs~\cite{GanianHNOMR12}, it is unclear how to use just this result to design any kind of canonical decompositions for classes of bounded shrubdepth. Consequently, so far it was unknown whether the graph isomorphism problem can be solved in fixed-parameter time on classes of bounded shrubdepth\footnote{This statement is somewhat imprecise, as shrubdepth is defined as a parameter of a graph class, rather than of a single graph. By this we mean that there is a universal constant $c$ such that for every graph class $\Cc$ of bounded shrubdepth, the isomorphism of graphs from $\Cc$ can be tested in $\Oh(n^c)$ time, with the constant hidden in the $\Oh(\cdot)$ notation possibly depending on $\Cc$.}. The most efficient isomorphism test in this context is the one designed by Grohe and Schweitzer~\cite{GroheS15} for the cliquewidth parameterization: it works in $\mathsf{XP}$ time, that is, in time $n^{f(k)}$ where $k$ is the cliquewidth and $f$ is a computable function. See also the later work of Grohe and Neuen~\cite{GroheN21}, which improves the $\mathsf{XP}$ running time and applies to the more general canonization problem.

We show that the Finitary Substitute Lemma can be used to bridge this gap by proving the following result.

\begin{restatable}{theorem}{mainshrubdepth}\label{thm:canonization}
	Let $\CC$ be a class of graphs of bounded shrubdepth.
	Then there is a class $\DD$ of binary structures of bounded treedepth and a mapping $\Aa\colon \CC\to \DD$ such that:
	\begin{itemize}
	 \item For each $G\in \CC$, the vertex set of $G$ is contained in the domain of $\Aa(G)$ and the mapping $G\mapsto \Aa(G)$ is isomorphism-invariant.
	 \item Given an $n$-vertex graph $G\in \CC$, the structure $\Aa(G)$ has $\Oh(n)$ elements and can be computed in time $\Oh(n^2)$.
	 \item There is a simple first-order interpretation $\Is$ such that $G=\Is(\Aa(G))$, for every $G\in \CC$.
	\end{itemize}

\end{restatable}

Here, by isomorphism-invariance we mean that every isomorphism between $G,G'\in \Cc$ extends to an isomorphism between $\Aa(G)$ and $\Aa(G')$.
Further, by a {\em{simple}} interpretation we mean a first-order interpretation that is $1$-dimensional: vertices of $G$ are interpreted in single elements of $\Aa(G)$ (actually, every vertex is interpreted in itself). Thus, $\Aa(G)$ can be regarded as a canonical --- obtained in an isomorphism-invariant way --- sparse decomposition of $G$ that encodes $G$ faithfully and takes the form of a structure of bounded treedepth. We remark that certain logic-based sparsification procedures for classes of bounded shrubdepth were proposed in~\cite{ChenF20,GajarskyKNMPST20}, but these are insufficient for our applications, which we explain next.

The third point above together with the fact that $\Aa$ is isomorphism-invariant imply the following: for all $G,G'\in \Cc$, $G$ and $G'$ are isomorphic if and only if $\Aa(G)$ and $\Aa(G')$ are.
We can now combine Theorem~\ref{thm:canonization} with the approach of Bouland et al.~\cite{BoulandDK12} to give a fixed-parameter isomorphism test on classes of bounded shrubdepth.

\begin{restatable}{theorem}{gi}\label{thm:gi}
	For every graph class $\Cc$ of bounded shrubdepth there is an $\Oh(n^2)$-time algorithm that given $n$-vertex graphs $G,G'\in \Cc$, decides whether $G$ and $G'$ are isomorphic.

\end{restatable}

In fact, our algorithm solves also the general canonization problem, see Section~\ref{subsec:sd-isomorphism}.

We remark that the algorithm of Theorem~\ref{thm:gi} is non-uniform, in the sense that we obtain a different algorithm for every class $\Cc$. Despite the existence of parameters such as rankdepth~\cite{KwonMOW21} or SC-depth~\cite{GanianHNOMR12} that are suited for the treatment of single graphs and are equivalent in terms of boundedness on classes to shrubdepth, we do not know how to make our algorithm uniform even for the rankdepth or SC-depth parameterizations.

Finally, we believe that the construction behind our proof of Theorem~\ref{thm:canonization} can be used to obtain an alternative proof of a result of Hlin\v{e}n\'y and Gajarsk\'y~\cite{GajarskyH12}, later reproved by Chen and Flum~\cite{ChenF20}: the expressive power of first-order and monadic second-order logic coincide on classes of bounded shrubdepth. This direction will be explored in future research.

%
%
%
%

%
%

\subparagraph*{Nowhere dense and monadically stable classes.}
Second, we use the Finitary Substitute Lemma to provide canonical strategies in game characterizations of two important concepts in structural graph theory: nowhere dense classes and monadically stable classes. In both cases, a strategy in the game can be regarded as decompositions of the graph in question.

We start with some definitions.
A {\em{unary lift}} of a class of graphs $\CC$ is any class of structures $\CC^+$ such that every member of $\CC^+$ is obtained from a graph belonging to $\CC$ by adding any number of unary predicates on vertices. A class of graphs $\CC$ is {\em{monadically stable}} if every unary lift of $\CC$ is stable. On the other hand, a class of graphs $\CC$ is {\em{nowhere dense}} if for every $d\in \N$ there exists $t$ such that no graph in $\CC$ contains the $d$-subdivision of $K_t$ as a subgraph.

Nowhere denseness is the most fundamental concept of uniform sparsity in graphs considered in the theory of Sparsity; see the monograph of Ne\v{s}et\v{r}il and Ossona de Mendez~\cite{sparsity} for an introduction to this area. A pinnacle result of this theory was derived by Grohe et al.~\cite{GroheKS17}: the model-checking problem for first-order logic is fixed-parameter tractable on any nowhere dense graph class. As observed by Adler and Adler~\cite{AdlerA14} using earlier results of Podewski and Ziegler~\cite{podewski1978stable}, monadically stable classes are dense counterparts of nowhere dense classes in the following sense: every nowhere dense class is monadically stable, and nowhere denseness and monadic stability coincide when we assume the class to be sparse, for instance to exclude some biclique $K_{t,t}$ as a subgraph. This motivated the following conjecture~\cite{warwick}, which is an object of intensive studies for the last few years: The model-checking problem for first-order logic is fixed-parameter tractable on every monadically stable class pf graphs $\Cc$.

To approach this conjecture, it is imperative to obtain a better structural understanding of graphs from monadically stable classes. This is the topic of several very recent works~\cite{Braunfeld_2021,braunfeld2022existential,abs-2209-11229,dreier2022indiscernibles,flipper_game}. In this work we are particularly interested in the results of Gajarsk\'y et al.~\cite{flipper_game}, who characterized monadically stable classes of graphs through a game model called the {\em{Flipper game}}, which reflects the characterization of nowhere dense classes through the {\em{Splitter game}}, due to Grohe et al.~\cite{GroheKS17}.

The radius-$r$ Splitter game is played on a graph $G$ between two players: {\em{Splitter}} and {\em{Connector}}. In every round, Connector first chooses any vertex $u$ and the current {\em{arena}} --- graph on which the game is played --- gets restricted to a ball of radius $r$ around $u$. Then Splitter removes any vertex of the graph. The game finishes, with Splitter's win, when the arena becomes empty. Splitter's goal is to win the game as quickly as possible, while Connector's goal is to avoid losing for as long as possible. The {\em{Flipper game}} is defined similarly, except that the moves of {\em{Flipper}} --- who replaces Splitter --- are as follows. Instead of removing a vertex, Flipper selects any subset of vertices $F$ and performs a {\em{flip}}: replaces all edges with both endpoints in $F$ with non-edges, and vice versa. Also, the game finishes when the arena consists of one vertex.

Grohe et al.~\cite{GroheKS17} proved that a class of graphs $\CC$ is nowhere dense if and only if for every radius $r\in \N$ there exists $\ell\in \N$ such that on every graph from $\CC$, Splitter can win the radius-$r$ Splitter game within at most $\ell$ rounds. This characterization is the backbone of their model-checking result for nowhere dense classes, as a strategy in the Splitter game provides a shallow decomposition of the graph in question, useful for understanding its first-order properties. Very recently, Gajarsk\'y et al.~\cite{flipper_game} proved an analogous characterization of monadically stable classes in terms of the Flipper game, and subsequently Dreier et al.~\cite{strnd} used this characterization to prove fixed-parameter tractability of the model-checking first order logic on monadically stable classes of graphs which possess so-called sparse neighborhood covers. Given this state-of-the-art, it is clear that a better understanding of strategies for Splitter and Flipper in the respective games may lead to a deeper insight into decompositional properties of nowhere dense and monadically stable graph classes.

In the Splitter game, we prove using just basic compactness, that in any arena there is only a bounded number of possible Splitter's moves that are {\em{progressing}}: lead to an arena where the Splitter can win in one less round. (See Theorem~\ref{thm:progressive} for a formal statement.) So this gives a transparent canonical strategy for Splitter: just play all progressive moves one by one, in any order.
\iffull For Flipper game, we refrain from giving a formal statement at this point (it can be found as Theorem~\ref{thm:definable_strategies}), as the concrete notion of definability is somewhat technical and requires further preliminaries; also, we actually work with a technical variant of Flipper game called {\em{Separator game}}, which was introduced in~\cite{flipper_game}.
But in essence, again the candidate next moves can be defined using a finitary first-order formula that works on the initial graph $G$ supplied with the previous moves of the Connector, and playing all those moves at once leads to a win within a bounded number of rounds.
\fi


%% file: preliminaries.tex
\section{Preliminaries}
\label{sec:preliminaries}


\subparagraph*{Models.}
We work with first-order logic over a fixed signature $\Sigma$ that consists of (possibly infinitely many) constant symbols and relation symbols. 
A \emph{model} is a $\Sigma$-structure, and is typically denoted $\str M,\str N$, etc. We usually do not distinguish between a model and its domain, when writing, for instance, $m\in\str M$ or $X\subset \str M$.
A graph $G$ is viewed as a model over the signature consisting of one binary relation denoted $E$, indicating adjacency between vertices.

If $\tup x$ is a finite set of variables, then we write $\phi(\tup x)$ to denote a first-order formula $\phi$ with free variables contained in $\tup x$. We may also write $\phi(\tup x_1,\ldots,\tup x_k)$ to denote a formula 
whose free variables are contained in $\tup x_1\cup\ldots\cup \tup x_k$.
We will write $x$ instead of $\set x$
in case of a singleton set of variables, e.g. 
$\phi(x,y)$ will always refer to a formula with two free variables $x$ and $y$.
We sometimes write $\phi(\tup x;\tup y)$ to distinguish a partition of the set of free variables of $\phi$ into two parts, $\tup x$ and $\tup y$; this partition plays an implicit role in some definitions.
A $\Sigma$-formula $\phi(\tup x)$ \emph{with parameters} from a set $A\subset \str M$ is a formula $\phi(\tup x)$ over the signature $\Sigma\uplus A$, where the elements of $A$ are treated as constant symbols (which are interpreted by themselves).

If $U$ is a set and $\tup x$ is a set of variables, then $U^{\tup x}$ denotes the set of all \emph{$\tup x$-tuples} $\tup a\from \tup x\to U$ of $\tup x$ in $U$.
For a formula $\phi(\tup x)$ (with or without parameters) and an $\tup x$-tuple $\tup m\in \str M^{\tup x}$, we write $\str M\models \phi(\tup m)$ if the valuation $\tup m$ satisfies the formula $\phi(\tup x)$ in $\str M$.
For a formula $\phi(\tup x; \tup y)$ and a tuple $\tup b \in \str M^{\tup y}$ we denote by $\phi(\str M^{\tup x}; \tup b)$ the set of all $\tup a \in \str M^{\tup x}$ such that $\str M \models \phi(\tup a; \tup b)$.

\subparagraph*{Theories and compactness.}
A \emph{theory} $T$ (over $\Sigma$) is a set of $\Sigma$-sentences.
The theory of a class of structures $\Cc$ is the set of sentences that hold in every structure $\str M \in \Cc$.
For instance, the theory of a class of graphs $\Cc$ contains sentences expressing that the relation $E$ is symmetric and irreflexive.
A model of a theory $T$ is a structure $\str M$ such that $\str M\models \phi$ for all $\phi\in T$. When a theory has a model, it is said to be \emph{consistent}.

\begin{theorem}[Compactness]\label{thm:compactness}
  A theory $T$ is consistent if and only if every finite subset $T'$ of $T$ is consistent.
\end{theorem}

\subparagraph*{Elementary extensions.}
Let $\str M$ and $\str N$ be two structures with $\str M\subset \str N$, that is, the domain of $\str M$ is contained in the domain of $\str N$.
Then $\str N$ is an \emph{elementary extension} of $\str M$, written $\str M\prec \str N$,
if for every formula $\phi(\tup x)$ (without parameters) and tuple $\tup m\in \str M^{\tup x}$, the following equivalence holds:
\[\str M\models \phi(\tup m)\qquad\textrm{if and only if}\qquad \str N\models \phi(\tup m).\]
\noindent We also say that $\str M$ is an \emph{elementary substructure} of $\str N$.
In other words, $\str M$ is an elementary substructure of $\str N$ if $\str M$ is an induced substructure of $\str N$, where we imagine that $\str M$ and 
$\str N$ are each equipped with every relation 
$R_\phi$ of arity $k$ (for $k\in\N$) that is defined by any fixed first-order formula $\phi(x_1,\ldots,x_k)$. 
In this intuition, formulas of arity $0$ correspond to Boolean flags, with the same valuation for both $\str M$ and $\str N$.

\subparagraph*{Interpretations and transductions.} An \emph{interpretation} $\Is$ between signatures $\Sigma$ and $\Gamma$ is specified by a domain formula $\delta(x)$ and a formula $\alpha_R(x_1,\dots,x_k)$ for each relation symbol $R \in \Gamma$ of arity $k$, with $\delta$ and the $\alpha_R$'s being in the signature $\Sigma$.
For a given $\Sigma$-structure $\str M$, the interpretation outputs the $ \Gamma$-structure $\Is(\str M)$ whose domain is $\delta(\str M)$ and in which the interpretation of each relation $R$ of arity $k$ consists of the tuples $\tup m$ such that $\str M \models \alpha_R(\tup m)$.

For an integer $k \in \N$ and a structure $\str M$, we define $k \times \str M$ to be the structure consisting of $k$ disjoint copies of $\str M$, together with a new symmetric binary relation $S$ containing all pairs $(m,m')$ such that $m$ and $m'$ originate from the same element of $\str M$.
A \emph{transduction} from $\Sigma$ to $\Gamma$ consists of an integer $k$, unary symbols $U_1,\dots,U_\ell$ and an interpretation $\Is$ from $\Sigma \cup \{S,U_1,\dots,U_\ell\}$ to $\Gamma$.

For a transduction $\Ts$ and an input $\Sigma$-structure $\str M$, the output $\Ts(\str M)$ consists of all $\Gamma$-structures $\str N$ such that there exists a coloring $\wh{\str M}$ of $k \times \str M$ with fresh unary predicates $U_1,\dots,U_\ell$ such that $\str B = \Is(\wh{\str M})$.
We say that a class of $\Sigma$-structures $\Cc$ {\em{transduces}} a class of $\Gamma$-structure $\Dd$  if there exists a transduction $\Ts$ such that for every structure $\str N \in \DD$ there is $\str M \in \Cc$ satisfying $ \str N \in \Ts(\str M)$.

\subparagraph*{Graphs.} We use standard graph theory notation. For a graph parameter $\pi$, we say that a graph class $\Cc$ has {\em{bounded $\pi$}} if there exists $k\in \N$ such that $\pi(G)\leq k$ for all $G\in \Cc$. Similarly, a class of structures $\Cc$ has bounded $\pi$ if the class of Gaifman graphs of structures in $\Cc$ has bounded $\pi$. 

%% file: splitter.tex
\section{Canonical Splitter-strategies in nowhere dense graphs}\label{sec:splitter}

In this section, we show how compactness can be used to derive canonical decompositions for nowhere dense classes.
More precisely, we will show that in the Splitter game, which characterizes nowhere dense classes~\cite{GroheKS17}, there is a constant $k$ (depending only on the graph class $\Cc$) such that for any graph in $\Cc$ there are at most $k$ optimal Splitter moves.
This will allow us to illustrate the general methodology used in the paper.

\subparagraph*{Splitter game.} First, we recall the rules of the Splitter game.
The radius-$r$ Splitter game is played on a graph $G$ by two players, Splitter and Connector, in rounds $i=1,2,\dots$ as follows.
Initially the arena $G_1$ is the whole graph $G$.
In the $i$-th round,
\begin{itemize}
\item Connector chooses a vertex $c_i \in G_i$;
\item Splitter chooses a vertex $s_i \in G_i$ and we let $G_{i+1}=G_i[B^r_{G_i}(c_i)] - s_i$;
\item Splitter wins if $G_{i+1}$ is the empty graph, otherwise the game continues.
\end{itemize}
Here, $B_H^r(u)=\{v\in V(H)~|~\dist_H(u,v)\leq r\}$ denotes the ball of radius $r$ around $u$ in $H$.

The following result is instrumental in the celebrated proof of model-checking on nowhere dense classes~\cite{GroheKS17}. 
\begin{theorem}[Theorems 4.2 and 4.5 in~\cite{GroheKS17}]
A class of graphs $\Cc$ is nowhere dense if and only if for every $r$, there exists $\ell$ such that on every graph $G \in \Cc$, Splitter can win the radius-$r$ game in at most $\ell$ rounds. 
\end{theorem}

The \emph{$r$-Splitter number} of a graph $G$ is the minimal $\ell$ such that Splitter wins the radius-$r$ game in $\ell$ rounds.
Fix a nowhere dense class $\Cc$ and a radius $r$, and let $\ell$ be as in the theorem (hence $\ell$ is an upper bound to all $r$-Splitter numbers of graphs in $\Cc$).
Observe that for a given $\ell' \leq \ell$ there is a first-order sentence expressing that Splitter wins the radius-$r$ game in $\leq \ell'$ rounds, and therefore, there is a first-order sentence expressing that $G$ has Splitter number $\ell'$.
Given a Connector move $c \in V(G)$, we say that a Splitter move $s \in V(G)$ is \emph{$r$-progressing against $c$} if the $r$-Splitter number of $G[B_r(c)] - s$ is strictly smaller than the $r$-Splitter number of $G[B_r(c)]$.
In other words, playing $s$ is strictly better for Splitter than not playing any vertex.
Again, since an upper bound to Splitter numbers depends only on~$\Cc$, this can be expressed by a formula $\phi_r(s;c)$.
This leads to the following result.

\begin{theorem}\label{thm:progressive}
Let $\Cc$ be a nowhere dense class of graphs, and $r \in \N$.
There is a constant $k$ such that for every graph $G \in \Cc$, and every Connector move $c$, there are at most $k$ progressing moves against $c$ in $G$.
\end{theorem}

In particular, this gives an isomorphism-invariant strategy for Splitter: simply play all progressing moves (either one by one, in any order, or all at once in an extended variant of the game considered in~\cite{GroheKS17}, where Splitter can remove a bounded number of vertices in each turn, instead of just one.)
The idea of the proof is to extend, by compactness, progressive moves towards outside the model (in an elementary extension), and conclude by observing that ``being a progressive move'' is a definable and hereditary property.

\begin{proof}
Let $T$ be the theory of $\Cc$.
Note that $T$ contains the sentence ``Splitter wins the radius-$r$ game in $\leq \ell$ rounds''. Our aim is to prove that for some $k$, it contains the sentence ``for all connector moves $c$, there are at most $k$ progressing Splitter moves against $c$''.
We show that for any model of $T$ and any connector move $c$, there are finitely many progressing Splitter moves against $c$; the result then follows from an easy application of compactness.

Assume for contradiction that there is a model $\str M$ of $T$ and a connector move $c \in \str M$ such that Splitter has infinitely many progressing moves against $c$.
We now let $T'$ be the theory over the signature extended by a constant corresponding to each element $m \in \str M$ and an additional constant $s$, such that $T'$ consists of:
\begin{itemize}
    \item all sentences in $T$,
    \item all sentences (with parameters in $\str M$) which hold in $\str M$,
    \item a sentence expressing that $s$ is a progressing move against $c$, and
    \item for each $m \in \str M$, a sentence expressing that $s \neq m$.
\end{itemize}
Since every finite subset $T''$ of $T'$ mentions finitely many $m \in \str M$, one can construct a model of $T''$ by starting from $\str M$ and setting $s$ to be one of those progressing moves that are not mentioned. 
We conclude from compactness (Theorem~\ref{thm:compactness}) that $T'$ is consistent.

Let $\str N$ be a model of $T'$.
By construction $\str N$ is an elementary extension of $\str M$ --- in particular, $\str N[B_r(c)]$ has the same Splitter number $\ell'$ as $\str M[B_r(c)]$ --- and contains a progressing move $s \in \str N \setminus \str M$ against $c$.
This means that $\str N[B_r(c)] - s$ has Splitter number $< \ell'$.
But $\str M[B_r(c)]$ is a subgraph of $\str N[B_r(c)] - s$ with Splitter number $\ell'$: this is absurd.
\end{proof}

The next section presents more elaborate tools from stability theory that will allow us to extend the above idea to different settings.

%% file: stability.tex
\section{Stability, forking, and Finitary Substitution}\label{sec:stability}

This section collects notions and a few basic results from stability theory.
The purpose is to give a self-contained exposition culminating in our Finitary Substitution Lemma; for more context and explanations we refer to~\cite{tent_ziegler}.

\subsection{Stability and definability of types}

We say that a formula $\phi(\tup x;\tup y)$ defines a {\em{ladder}} of order $k$ \iffull (see Figure \ref{fig:ladder}) \else \fi in a model $\str M$ if there are sequences $\tup a_1,\ldots\tup a_{k}\in \str M^{\tup x}$ and $\tup b_1,\ldots,\tup b_k\in\str M^{\tup y}$ satisfying \[\str M\models\phi(\tup a_i;\tup b_j)\qquad\textrm{if and only if}\qquad i < j,\qquad\text{for $1\le i,j\le k$}.\]
For a formula $\phi(\tup x; \tup y)$ we call the largest $k$ such that $\phi$ defines a ladder of order $k$ the \emph{ladder index} of $\phi$ in $\str M$.
If no such $k$ exists, we say that the ladder index of $\phi$ is $\infty$.

\iffull
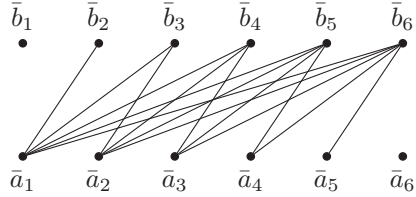
\begin{figure}[h]
    \centering
    
    \begin{tikzpicture} 
        \tikzstyle{vertex} = [draw, circle, fill, inner sep=1pt]
    
        \foreach \i in {1,2,3,4,5,6} {
        \node[vertex] (a\i) at (\i,0) [label=below:$\tup a_\i$] {};
        \node[vertex] (b\i) at (\i,1.5) [label=above:$\tup b_\i$] {};
        }
        \foreach \i/\j in {1/2,1/3,1/4,1/5,1/6,2/3,2/4,2/5,2/6,3/4,3/5,3/6,4/5,4/6,5/6} {
        \draw (a\i) -- (b\j);
        }
    \end{tikzpicture}
  
    \caption{A ladder of order $6$.}
    \label{fig:ladder}
\end{figure}
\else \fi

We say that $\phi$ is \emph{stable} in $\str M$ if its ladder index is finite.
We say that $\phi$ is stable in a theory $T$ if it is stable in all models of $T$.
Moreover, we say that a model (or a theory) is stable if every formula is stable.

We now state a fundamental result about stable formulas; it states that sets definable by stable formulas with parameters in some elementary extension can actually be defined from the model itself.

\begin{theorem}[Definability of types]\label{thm:definability-of-types}
  Let $\str M\prec \str N$ be two models and $\phi(\tup x; \tup y)$ be a stable formula of ladder index $d$ in $\str M$.
  For every $\tup n \in \str N^{\tup y}$ there is a formula $\psi(\tup x)$, which is a positive boolean combination of formulas of the form $\psi(\tup x; \tup m)$ using $2 d +1$ parameters $\tup m \in \str M^{\tup y}$, such that for every $\tup a \in \str M^{\tup x}$,
  \[
    \str N \models \phi(\tup a; \tup n) \quad \textrm{ if and only if } \quad \str M \models \psi(\tup a).
  \]
\end{theorem}

\subsection{Forking in stable theories}

We move on to the definition of forking, which was first defined by Shelah in order to study stable theories~\cite{classification_theory_v1}, and later grew to become the central notion of stability theory.
In stable theories, forking coincides with the simpler notion of dividing, so by a slight abuse we will only work with dividing (and call it forking).
We first need to formally introduce types, then we give a definition of forking in stable theories and a few useful properties.

\subparagraph*{Types.}
Fix a model $\str M$ over a signature $\Sigma$.
A set $\pi$ of formulas in variables $\tup x$ with parameters from $A \subseteq \str M$ is called a \emph{partial type over $A$} if it is \emph{consistent}: for every finite subset $\pi' \subseteq \pi$ there is $\tup m \in \str M^{\tup x}$ which satisfies all the formulas from $\pi'$ (i.e. for every formula $\phi(\tup x) \in \pi'$ we have $\str M \models \phi(\tup m)$).
We sometimes write $\pi(\tup x)$ to explicitly mention free variables.
Partial types $p$ which are maximal are called \emph{types}; this amounts to stating that for every formula $\phi(\tup x)$ with parameters from $A$, either $\phi(\tup x) \in p$ or $\neg \phi(\tup x) \in p$.
Observe that for sets $A \subseteq B \subseteq \str M$ every type $p$ over $A$ can be seen as a partial type over $B$.
We denote the set of types over $A$ in variables $\tup x$ by $S_{\tup x}(A)$.

For a tuple $\tup a \in \str M^{\bar x}$ and a set $A \subset \str M$ of parameters, the type of $\tup a$ over $A$, denoted $\tp(\tup a / A) \in S_{\tup x}(A)$, is the set of all formulas $\phi(\tup x)$ with parameters from $A$ such that $\str M \models \phi(\tup a)$.
It follows from compactness that for every $p \in S_{\tup x}(\str M)$ there is some $\str N \succ \str M$ and an $\tup x$-tuple $\tup n \in \str N^{\tup x}$ such that $\tp(\tup n / \str M) = p$.

\subparagraph*{Forking.}
Fix a stable model $\str M$ over a signature $\Sigma$ and a set $A \subset \str M$.
Let $\phi(\tup x; \tup y)$ be a formula without parameters and let $\tup b \in \str M^{\tup y}$.
We say that $\phi(\tup x; \tup b)$ \emph{forks over $A$} if there is an elementary extension $\str N \succ \str M$, a sequence $\tup b_1, \tup b_2, \ldots \in \str N^{\bar y}$ satisfying $\tp(\tup b_i / A) = \tp(\tup b / A)$ for every $i$ and an integer $k$ such that $S = \set{\phi(\tup x; \tup b_i): i \in \N}$ is $k$-inconsistent: no $k$-element subset of $S$ is consistent.
For a type $p \in S_{\tup x}(B)$ over a set $B \subseteq \str M$, we say that $p$ forks over $A$ if there is a formula $\phi(\tup x;\tup b) \in p$ which forks over $A$.

We will make use of the following important property of forking which is often called \emph{(full) existence}.

\begin{theorem}[{See \cite[Corollary 7.2.7]{tent_ziegler}}]
    \label{thm:existence-for-types}
    Let $\str M$ be a stable model and let $A \subseteq B \subset \str M$.
    For every $p \in S_{\tup x}(A)$ there is some $q \in S_{\tup x}(B)$ such that $p \subseteq q$ and $q$ does not fork over $A$.
\end{theorem}

\subparagraph*{Finitary formulas.}
We say that a formula $\phi(\tup x; \tup y)$ is \emph{finitary} in a theory $T$ if for every model $\str M$ of $T$, the set $\{\phi(\str M^{\tup x}; \tup m)\, \colon\, m \in \str M^{\tup y}\}$ is finite. By compactness, this is equivalent to the following assertion: there exists $k\in \N$ such that $\left|\{\phi(\str M^{\tup x}; \tup m)\, \colon\, m \in \str M^{\tup y}\}\right|\leq k$ for every model $\str M$ of $T$.
We now relate forking and finitary formulas.

\begin{theorem}[Special case of {\cite[Theorem 8.5.1]{tent_ziegler}}\protect\footnotemark]
    \footnotetext{Formally, \cite[Theorem 8.5.1]{tent_ziegler} speaks about definability with \emph{imaginaries}, which is known to be equivalent to the existence of finitary formulas (see for instance {\cite[Lemma 1.3.2 (5), Lemma 1.3.7]{chernikov_lecture_notes}}).}
    \label{thm:definable-forking}
    Let $\str M$ be a stable model, $\str N$ an elementary extension of $\str M$, $\phi(\tup x; \tup y)$ a formula, $\tup n \in \str N^{\tup y}$, and $A\subseteq \str M$.
    If $\tp(\tup n / \str M)$ does not fork over $A$, then there is a finitary formula $\phi'(\tup x; \tup z)$ and a tuple $\tup r \in \str M^{\tup z}$ such that $\phi(\str N^{\tup x}; \tup n) \cap \str M^{\tup x} = \phi'(\str M^{\tup x}; \tup r)$.
\end{theorem}

Combining \Cref{thm:existence-for-types,thm:definable-forking} yields the following statement.

\begin{lemma}\label{lem:mt-core}
Let $\str M$ be a stable model over the signature $\Sigma$, $\phi(\tup x;\tup y)$ a $\Sigma$-formula, and $\psi$ a sentence over signature $\Sigma \cup \{R\}$, where $R \notin \Sigma$ has arity $|\tup x|$.
Let $\tup s \in \str M^{\tup y}$ be such that $\str M \models \psi[R(\tup x) / \phi(\tup x;\tup s)]$.
Then there is an elementary extension $\str N$ of $\str M$, a tuple $\tup s' \in \str N^{\tup y}$, a finitary formula $\phi'(\tup x; \tup z)$ and a tuple $\tup r \in \str M^{\tup z}$, such that $\str N \models \psi[R(\tup x) / \phi(\tup x; \tup s')]$ and $\phi(\str N^{\tup x};\tup s') \cap \str M^{\tup x} = \phi'(\str M^{\tup x}; \tup r)$.
\end{lemma}

\begin{proof}
Consider $p = \tp(\tup s / \emptyset)$.
By \Cref{thm:existence-for-types}, $p$ extends to a type $q \in S_{\tup y}(\str M)$ which does not fork over $\emptyset$.
By compactness there is an elementary extension $\str N \succ \str M$ and a tuple $\tup s' \in \str N^{\bar y}$ such that $\tp(\tup s' / \str M) = q$.
In particular $\tp(\tup s' / \emptyset) = p = \tp(\tup s / \emptyset)$, and therefore $\str N \models \psi[R(\tup x) / \phi(\tup x; \tup s')]$ as required.
Applying Theorem~\ref{thm:definable-forking} we get a finitary formula $\phi'(\tup x; \tup z)$ and a tuple $\tup r \in \str M^{\tup z}$ with the wanted properties.
\end{proof}

\subsection{Finitary Substitute Lemma}

Recall from Section~\ref{sec:splitter} that applying our method requires a mechanism for moving the wanted property $\psi$ back towards the structure $\str M$ we started from.
This is formalized by the following definition.
In a theory~$T$, and given two sentences $\psi$ and $\psi'$ over the signature $\Sigma \cup \{R\}$, we say that a sentence $\psi$ \emph{induces $\psi'$ on semi-elementary substructures} if for every model $\str M$ of~$T$, for every elementary extension $\str N$ and for every $\Rr \subseteq \str N^{k}$, where $k$ is the arity of $R$,
\[
    \str N[R / \Rr] \models \psi \qquad \textrm{implies} \qquad \str M[R / \Rr|_{\str M}] \models \psi'.
\]
As an important special case, if $\psi$ is hereditary then $\psi$ induces $\psi$ on semi-elementary substructures.
We are now ready to state our main model-theoretic tool.

\begin{lemma}[Finitary Substitute Lemma]\label{lem:finitary-substitute}
    Let $T$ be a theory with signature $\Sigma$, $\phi(\tup x;\tup y)$ a stable formula, and $\psi,\psi'$ be sentences over the signature $\Sigma \cup \{R\}$, where $R \notin \Sigma$ has arity $|\tup x|$, such that $\psi$ induces $\psi'$ on semi-elementary substructures.
    Assume that $T \models \exists \tup s. \psi[R(\tup x) / \phi(\tup x; \tup s)]$.
    Then there is a finitary formula $\phi'(\tup x; \tup z)$ such that $T \models \exists \tup s. \psi'[R(\tup x) / \phi'(\tup x; \tup s)]$.
\end{lemma}

\iffull
\begin{proof}
Define $\Delta_n$ to be the set of finitary formulas of size $\leq n$.
Assume towards contradiction that for all $n$, there is a model $\str M_n$ of $T$ such that for all $\phi' \in \Delta_n$, $\str M_n \not\models \exists \tup s. \psi'[R(\tup x) / \phi'(\tup x; \tup s)]$.
Then by compactness there is a model $\str M$ of $T$ such that for all finitary formulas, $\str M \not\models \exists \tup s. \psi'[R(\tup x) / \phi'(\tup x; \tup s)]$.
However, since $\str M$ is a model of $T$, our assumption gives us $\tup s \in \str M^{\tup y}$ such that $\str M \models \psi[R(\tup x) / \phi(\tup x; \tup s)]$, so we may apply Lemma~\ref{lem:mt-core}.

This gives an elementary extension $\str N$ of $\str M$, a tuple $\tup s' \in \str N^{\tup y}$, a finitary formula $\phi'(\tup x; \tup z)$ and a tuple $\tup r \in \str M^{\tup z}$ such that $\str N \models \psi[R(\tup x) / \phi(\tup x; \tup s')]$ and $\phi'(\str M^{\tup x};\tup r) = \phi(\str N^{\tup x}; \tup s') \cap \str M^{\tup x}$.
Stated differently, the interpretation $\Rr = \phi(\str N^{\tup x}; \tup s')$ of $R$ over $\str N$ satisfies $\str N \models \psi[R / \Rr]$ and $\Rr|_{\str M}(\tup x) = \phi(\str N^{\tup x}; \tup s') \cap \str M^{\tup x} = \phi'(\str M^{\tup x};\tup r)$.
Since $\psi$ induces $\psi'$ on semi-elementary substructures, we get that $\str M \models \psi'[R(\tup x) / \phi'(\tup x; \tup r)]$, a contradiction.

Therefore, there is $n$ such that for every model $\str M$ of $T$, there is $\phi' \in \Delta_n$ with $\str M \models \exists \tup s. \psi'[R(\tup x) / \phi'(\tup x; \tup s)]$.
Enumerate $\Delta_n=\{\phi'_0,\phi'_1,\dots,\phi'_{k-1}\}$ and take $\tup t=t_0t_1 \dots t_\ell$ to be a tuple of variables with $\ell = \lceil \log k \rceil$.
For each $i \in \{0,\dots, k-1\}$, define a formula
\[
    \enc_i(\tup t) = \bigwedge_{j=1}^{\ell} \begin{cases}
        t_j = t_0 \text{ if the } j \text{-th bit in }i \text{'s binary encoding is 0} \\
        t_j \neq t_0 \text{ otherwise},
    \end{cases}
\]
which declares that $\tup t$ encodes the number $i$.
Then let
\[
    \phi''(\tup x; \tup z \tup t) = \bigvee_{i=0}^{k-1} [\enc_i(\tup t) \wedge \phi'_i(\tup x; \tup z)].
\]
Note that $\phi''$ is finitary since each $\phi'_i$ is, and for any model $\str M$ of $T$,
\[
    \{\phi''(\str M^{\tup x}; \tup m) \colon \tup m \in \str M^{\tup z\tup t}\} \setminus \set{\emptyset}= \bigcup_i \ \{\phi_i'(\str M^{\tup x}; \tup m) \colon \tup m \in \str M^{\tup z}\}.
\]
Let us now assume without loss of generality that models of $T$ have $\geq 2$ elements (models of size $1$ can be dealt with by hardcoding the choice of $\phi'_i$ in $\phi''$ according to which relations hold).
We now show that $\phi''$ satisfies the conclusion of the Corollary: $\exists \tup s. \psi'[R(\tup x) / \phi''(\tup x; \tup s)] $ is entailed in all models of $T$.

Let $\str M$ be a model of $T$, and let $i \in \{0,\dots, k-1\}$ and $\tup s \in \str M^{\tup z}$ be such that $\str M \models \psi'[R(\tup x) / \phi'_i(\tup x; \tup s)]$.
Let $b_0$ and $b_1$ be any two different elements in $\str M$ and let $\tup b$ be the unique tuple in $\{b_0,b_1\}^{\tup t}$ such that $\str M \models \enc_i(\tup b)$.
Then $\str M \models \psi'[R(\tup x) / \phi''(\tup x; \tup s \tup b)]$ which yields the wanted result.
\end{proof}
\fi

%% file: canonisation.tex
\section{Canonization of graphs of bounded shrubdepth}\label{sec:shrubdepth}

In this section, we prove \Cref{thm:canonization,thm:gi} which we now recall for convenience.

\mainshrubdepth*
\gi*

\medskip

The proof is broken into three parts.
\begin{itemize}
    \item The first part combines insights about classes of bounded shrubdepth with our Finitary Substitute Lemma developed in the previous section, to conclude that the first level in a shrubdepth decomposition (which we will call a dicing, defined below) can be defined using finitary formulas.
    This result is stated as Theorem~\ref{thm:shrub-inductive} below.
    \item The second part builds on Theorem~\ref{thm:shrub-inductive} to propose a canonical transformation from classes of bounded shrubdepth to classes of bounded treedepth. This proves Theorem~\ref{thm:canonization}.
    \item In the third part, we show how~\Cref{thm:gi} is derived from Theorem~\ref{thm:canonization}, and also establish a stronger result about the canonization problem.
\end{itemize}
We start by recalling a few preliminaries about shrubdepth in Section~\ref{subsec:sd-prelim}, and proceed with the three parts outlined above in Sections~\ref{subsec:sd-definable},~\ref{subsec:sd-construction} and~\ref{subsec:sd-isomorphism}.

\subsection{Preliminaries on shrubdepth}\label{subsec:sd-prelim}

\subparagraph*{Shrubdepth.}
The decomposition notion underlying shrubdepth is that of {\em{connection models}}, defined as follows. Let $G$ be a graph. A {\em{connection model}} for $G$ consists of:
\begin{itemize}
 \item a finite set of labels $\Labels$;
 \item a labelling $\lbl\colon V(G)\to \Labels$;
 \item a rooted tree $T$ whose leaf set coincides with the vertex set of $G$; and
 \item for every non-leaf node $x$ of $T$, a symmetric relation $\Adj(x)\subseteq \Labels\times \Labels$, called the {\em{adjacency table}} at $x$.
\end{itemize}
\iffull
The rule is as follows: for every distinct vertices $u,v$ of $G$, 
$$
\textrm{$u$ and $v$ are adjacent in $G$}\qquad\textrm{if and only if}\qquad (\lbl(u),\lbl(v))\in \Adj(x),$$
where $x$ is the lowest common ancestor of $u$ and $v$ in $T$. 
\fi

The {\em{depth}} of a connection model is the depth of $T$. The {\em{shrubdepth}} of a graph class $\Cc$ is the least integer $d$ with the following property: there exists a finite set of labels $\Labels$ such that every graph $G\in \Cc$ has a connection model of depth at most $d$ that uses label set $\Labels$.

\subparagraph*{Dicings.}
Our inductive proof requires manipulating the first level (just below the root) of a connection model; we will call this a dicing.
Formally, for a graph $G$, a pair $(\Pp,\Ll)$ of partitions of the vertex set of $G$ is called a {\em{dicing}} of $G$ if for every pair of vertices $u,v$ belonging to different parts of $\Pp$, whether $u$ and $v$ are adjacent in $G$ depends only on the pair of parts of $\Ll$ that $u$ and $v$ belong to. In other words, there is a symmetric relation $Z\subseteq \Ll\times \Ll$ such that for all $u,v$ belonging to different parts of $\Pp$,
$$\textrm{$u$ and $v$ are adjacent in $G$}\qquad\textrm{if and only if}\qquad (\Ll(u),\Ll(v))\in Z,$$
where $\Ll(w)$ denotes the part of $\Ll$ to which $w$ belongs. In the context of a dicing $(\Pp,\Ll)$, partition $\Pp$ will be called the {\em{component partition}}, and partition $\Ll$ will be called the {\em{label partition}}. The {\em{order}} of a dicing $(\Pp,\Ll)$ is $|\Ll|$, the number of parts in the label partition.

\iffull
Note that if $(\Pp,\Ll)$ is a dicing of a graph $G$ of order $\ell$ and $H$ is an induced subgraph of~$G$, then restricting $\Pp$ and $\Ll$ to $H$ yields a dicing of order $\leq \ell$ for $H$.
Likewise, refining the label partition $\Ll$ preserves the property of being a dicing (though, of course, not the order).
This is not true in general for the component partition $\Pp$. 
\else \fi
Dicings appear naturally in connection models for shrubdepth: given a connection model for a graph $G$, using ``having a common ancestor below the root'' as component partition $\Pp$ and $\lbl$-classes as label partition $\Ll$ defines a dicing of $G$.

\iffull
\subparagraph*{Closure under transductions.}
An important result about classes of bounded shrubdepth, which is also used in our proof, is that they are closed under first-order transductions~\cite{GanianHNOM19}.

\begin{theorem}[Theorem 4.7 in~\cite{GanianHNOM19}]\label{thm:sd-closed-under-transductions}
If a class of graphs $\Cc$ has bounded shrubdepth and a class of graphs $\Dd$ can be tranduced from $\Cc$, then $\Dd$ has bounded shrubdepth as well.
\end{theorem}
\else \fi

\subsection{Definability of canonical dicings}\label{subsec:sd-definable}

We say that a formula $\varphi(\tup x;\tup y)$ with $|\tup x|=2$ {\em{defines a partition}} if for every graph $G$ and $\tup b\in G^{\tup y}$, $\varphi(G^{\tup x};\tup b)$ is an equivalence relation on the vertex set of $G$. (Note that for different choices of $\tup b$, $\varphi$ can yield different equivalence relations.) Abusing the notation, by $\varphi(G^{\tup x};\tup b)$ we will also denote the partition of the vertex set into the equivalence classes of $\varphi(G^{\tup x};\tup b)$.
Recall that a formula $\phi(\tup x; \tup y)$ is said to be finitary in (the theory of) a graph class $\Cc$ if there exists $k$ such that for all graph $G \in \Cc$,
\[
    |\{\phi(G^{\tup x}; \tup b)\,\colon\,\tup b \in G^{\tup y}\}| \leq k.
\]

This section is focused on establishing the following result.

\begin{theorem}\label{thm:shrub-inductive}
    Let $\Cc$ be a class of shrubdepth at most $d$, where $d>1$. Then there exists a hereditary class $\Cc'$ of shrubdepth at most $d-1$, finitary first-order formulas $\varphi(\tup x;\tup y)$ and $\lambda(\tup x;\tup y)$, each defining a partition, and $\ell\in \N$, such that the following holds: for every graph $G\in \Cc$ there exists $\tup s\in G^{\tup y}$ such that
    \begin{itemize}
     \item $(\varphi(G^{\tup x};\tup s),\lambda(G^{\tup x};\tup s))$ is a dicing of $G$ of order at most $\ell$; and
     \item for every part $A$ of $\varphi(G^{\tup x};\tup s)$, we have $G[A]\in \Cc'$.
    \end{itemize}
\end{theorem}

On a high level, this proves that connection models can be defined using first-order formulas $\phi(\tup x; \tup y),\lambda(\tup x; \tup y)$ and parameters $\tup s \in G^{\tup y}$.
While a good start towards sparsification, this alone would be insufficient for our needs, as different choices of $\tup s$ may lead to many different connection models, and choosing an arbitrary $\tup s$ would not give an isomorphism-invariant construction.
This difficulty is overcome by the finitariness of $\phi$ and~$\lambda$: our construction will take into account all of the (boundedly many) possible dicings (see Section~\ref{subsec:sd-construction}).

The proof of Theorem~\ref{thm:shrub-inductive} is broken into three parts as follows.
\begin{itemize}
    \item The first part consists of proving that the label partition $\Ll$ can be chosen to be definable as a partition $\lambda(G^{\tup x}; \tup s)$ into $\tup s$-types.
    This is achieved thanks to a more general result of Bonnet et al.~\cite{boundedLocalCliquewidth} pertaining to classes of bounded VC-dimension.
    \item We then show that the component partition $\Pp$ can be chosen to be definable by a formula $\phi(G^{\tup x}; \tup s)$ using the same parameters $\tup s$.
    This relies on known properties of classes of bounded shrubdepth~\cite{GanianHNOM19}\iffull, including Theorem~\ref{thm:sd-closed-under-transductions}.\fi
    \item We then apply our Finitary Substitute Lemma (\Cref{lem:finitary-substitute}) and prove that $\phi$ and $\lambda$ can be taken to be finitary.
\end{itemize}

\subparagraph*{Definability of the label partition.}
For a subset of vertices $S$ of a graph $G$ we let $\Ll_S$ denote the partition of the vertex set of $G$ into neighborhood classes with respect to $S$: $u$ and $v$ belong to the same part of $\Ll_S$ if and only if 
\[
    \{w\in S~|~u\textrm{ and }w\textrm{ are adjacent}\}=\{w\in S~|~v\textrm{ and }w\textrm{ are adjacent}\}.
\]
Note that we have $|\Ll_S|\leq 2^{|S|}$.
It turns out that label partitions can be taken of this form.

\begin{restatable}[follows from Theorem 3.5 of~\cite{boundedLocalCliquewidth}]{lemma}{labelpartition}\label{thm:incremental}
 Let $\Cc$ be graph class of bounded shrubdepth. Then for every graph $G\in \Cc$ and dicing $(\Pp,\Ll)$ of $G$ of order at most $t$, there exists $S\subseteq V(G)$ with $|S|\leq \Oh(t^2)$ such that $(\Pp,\Ll_S)$ is also a dicing of $G$.
\end{restatable}

\iffull
We prove Lemma~\ref{thm:incremental} in Appendix~\ref{app:label-partition}.
\else \fi

\subparagraph*{Definability of the component partition.}
We now show that the component partition $\Pp$ can also be defined using a first-order formula.

\begin{lemma}\label{lem:def-dicing}
    Let $\Cc$ be a graph class of bounded shrubdepth and $t\in \N$ be an integer.
    There exist formulas $\varphi(\tup x;\tup y)$ and $\lambda(\tup x;\tup y)$, both defining a partition, such that the following holds: for every graph $G\in \Cc$ and dicing $(\Pp,\Ll)$ of $G$ of order at most $t$, there exists $\tup s\in G^{\tup y}$ such that
    $$(\Pp',\Ll')=\left(\varphi(G^{\tup x};\tup s),\lambda(G^{\tup x};\tup s)\right)$$
    is a dicing of $G$ of order at most $2^{\Oh(t^2)}$. Further, every part of $\Pp'$ is entirely contained in some part of~$\Pp$.
\end{lemma}

\iffull
   \begin{proof}
    Let $G \in \Cc$ and fix a dicing $(\Pp,\Ll)$ of $G$ order at most $t$.
    By \Cref{thm:incremental}, there exists a vertex subset $S$ with $|S|\leq \Oh(t^2)$ such that $(\Pp,\Ll_S)$ is also a dicing of $G$. In other words, there is a symmetric relation $Z\subseteq \Ll_S\times \Ll_S$ such that for vertices $u,v$ belonging to different parts of $\Pp$, we have
   \begin{equation}\label{eq:wydra}
   \textrm{$u$ and $v$ are adjacent in $G$}\qquad\textrm{if and only if}\qquad (\Ll_S(u),\Ll_S(v))\in Z.
   \end{equation}
   Since $|S|\leq \Oh(t^2)$, there exists a constant $c\in \Oh(t^2)$, depending on $\Cc$ and $t$ but not on the choice of $G \in \Cc$, such that $|S|\leq c$.
   Consequently, $|\Ll_S|\leq 2^c$.
   
   Let $H$ be the graph obtained from $G$ by {\em{flipping}} according to $Z$: the vertex sets of $H$ and $G$ coincide, and for distinct $u,v\in V(G)$ we have
   $$\textrm{$u,v$ are adjacent in $H$}\quad\textrm{if and only if}\quad \textrm{($u,v$ are adjacent in $G$)}\ \mathsf{xor}\ \textrm{(}(\Ll_S(u),\Ll_S(v))\in Z\textrm{)}.$$
   Note that by \eqref{eq:wydra}, we have the following.
   
   \begin{claim}\label{cl:comps-cont}
    Every edge of $H$ connects two vertices that are in the same part of $\Pp$. Therefore the vertex set of every connected component of $H$ is entirely contained in a single part of $\Pp$.
   \end{claim}
   
   Moreover, let us now prove the following.
   
   \begin{claim}
    There exists a constant $d\in \N$ depending only on $\Cc$ such that every connected component of $H$ has diameter at most $d$.
   \end{claim}
   \begin{claimproof}
    Let $\widehat{\Cc}$ be the class of graphs that can be obtained as follows.
    Take any $J\in \Cc$ and any partition $\Ll$ of the vertex set of $J$ into at most $2^c$ parts. Then for any symmetric $W\subseteq \Ll\times \Ll$, apply the flipping according to $W$ (as described above) to $J$ and include the obtained graph~in~$\widehat{\Cc}$.
    
    Clearly, we have $H\in \widehat{\Cc}$. Moreover, $\widehat{\Cc}$ can be  transduced from $\Cc$, and thus Theorem~\ref{thm:sd-closed-under-transductions} concludes that $\widehat{\Cc}$ also has bounded shrubdepth.
    As proved in~\cite{GanianHNOM19}, every class of bounded shrubdepth excludes some path as an induced subgraph. Since shortest paths are induced, it follows that the diameter of every connected component of $H$ cannot be larger than the length of the path excluded in $\widehat{\Cc}$.
   \end{claimproof}
   
   Let $\tup s\in G^{\tup y}$ be an arbitrary enumeration of $S$ as a tuple, where $\tup y$ is a tuple of $c$ variables; in case $|S|<c$, some of the elements of $\tup s$ are repeated.
   Since reachability by a path of length at most $d$ can be defined using a first-order formula, and the edge relation in $H$ can be interpreted from the edge relation in $G$ and the set $S$, we conclude the following.
   
   \begin{claim}\label{cl:H-conn-in-FO}
        There is a first-order formula $\varphi(x,x';\tup y)$ depending only on $\Cc$ and where $\tup y$ is a tuple of $t$ variables, such that for $u,v\in V(G)$, we have
        \[
            G\models \varphi(u,v;\tup s) \quad \textrm{if and only if}\quad u\textrm{ and }v\textrm{ belong to the same connected component of }H.
        \]
   \end{claim}
   
   Let $\lambda(x,x';\tup y)$ be the formula stating that $x$ and $x'$ have the same neighbors among the entries of $\tup y$. 
   Then we claim that formula $\varphi(x,x';\tup y)$ provided by Claim~\ref{cl:H-conn-in-FO} together with $\lambda(x,x';\tup y)$ satisfy all the required properties.
   First, $\lambda$ clearly defines a partition, and we may assume that so does $\varphi$. (This can be done by checking, within $\varphi$, whether the defined relation is an equivalence relation, and if this is not the case, yielding any equivalence relation instead, for instance the complete one.) Second, observe that $\varphi(G^{\tup x};\tup s)$ is the partition of the vertex set of $G$ into the connected components of $H$, while $\lambda(G^{\tup x};\tup s)$ is $\Ll_S$. Then:
   \begin{itemize}
        \item as in $H$ there are no edges between different connected components of $H$,  $(\varphi(G^{\tup x};\tup s),\Ll_S)$ is a dicing of $G$; and
        \item by Claim~\ref{cl:comps-cont}, every part of $\varphi(G^{\tup x};\tup s)$ is entirely contained in a single part of $\Pp$.
   \end{itemize}
   Finally, $|\Ll_S|\leq 2^c\leq 2^{\Oh(t^2)}$. This concludes the proof.
\end{proof}
\fi

\subparagraph*{Finitariness of the definition.}
\iffull Before moving on to the proof of Theorem~\ref{thm:shrub-inductive}, we need an additional insight which was observed by Ganian et al.~\cite{GanianHNOM19} using an earlier work of Ding~\cite{Ding92}.
Recall here that a graph class $\Cc$ is {\em{hereditary}} if it is closed under taking induced subgraphs.

\begin{theorem}[follows from Corollary~3.9 of~\cite{GanianHNOM19}]
 For every hereditary graph class $\Cc$ of bounded shrubdepth there exists a finite set $\Forb(\Cc)$ of graphs such that for every graph $G$,
 \[
    G\in \Cc \textrm{ if and only if } G\textrm{ does not contain any member of }\Forb(\Cc)\textrm{ as an induced subgraph.}
\]
\end{theorem}

From the above it immediately follows that classes of bounded shrubdepth are first-order definable in the following sense.

\begin{corollary}\label{cor:fo-def}
 For every hereditary graph class $\Cc$ of bounded shrubdepth there exists a first-order sentence $\psi_\Cc$ such that for every graph $G$,
 \[
    G\in \Cc \textrm{ if and only if } G\models \psi_\Cc.
 \]
\end{corollary}

With these tools in place, we are ready to derive the theorem.
\fi
\iffull
\begin{proof}[Proof of Theorem~\ref{thm:shrub-inductive}]
    Consider any $G\in \Cc$ and let $(\Labels,\lbl,T,\Adj)$ be a connection model for $G$, where $T$ has depth at most $d$ and $\Labels$ is a fixed label set depending only on $\Cc$.
    Let $r$ be the root of $T$ and let $\Pp$ be the partition of the vertex set of $G$ according to the subtrees of $T$ rooted at the children of $r$: vertices $u,v$ are in the same part of $\Pp$ if and only if there is a child $x$ of $r$ that is a common ancestor of $u$ and $v$.
    Further, let $\Ll$ be the label partition: vertices $u,v$ are in the same part of $\Ll$ if and only if $\lbl(u)=\lbl(v)$.
    By the definition of a connection model, $(\Pp,\Ll)$ is a dicing of $G$ of order at most $|\Labels|$.

    Let $\Cc'$ be the class of all graphs that admit a connection model of depth at most $d-1$ using the label set $\Labels$. Note that $\Cc'$ is hereditary and has shrubdepth at most $d-1$. Further, $G[A]\in \Cc'$ for every part $A$ of $\Pp$. 
   
    By \Cref{lem:def-dicing}, there exist formulas $\varphi(\tup x;\tup y)$ and $\lambda(\tup x;\tup y)$, depending only on $\Cc$, such that there is $\tup s\in G^{\tup y}$ for which $(\varphi(G^{\tup x};\tup s),\lambda(G^{\tup x};\tup s))$ is a dicing of $G$ of order at most $\ell$, where $\ell\in 2^{\Oh(|\Labels|^2)}$ is a constant depending only on $\Cc$. Moreover, every part of $\varphi(G^{\tup x};\tup s)$ is entirely contained in a single part of $\Pp$, which implies that for every part $A'$ of $\varphi(G^{\tup x},\tup s)$ we have $G[A']\in \Cc'$.
    There remains to transform $\phi$ and $\lambda$ into finitary formulas.

Let $R$ be a relation symbol of arity $4$ and consider the following assertion:
\begin{center}
``$R$ is the product of two partitions $\Pp$ and $\Ll$ such that $(\Pp, \Ll)$ is a dicing of $G$ of \\ order at most $\ell$.
Moreover, for every part $A$ of $\Pp$ it holds that~$G[A] \in \Cc'$''.
\end{center}
We claim that it can be expressed by a first order sentence $\psi$ over the signature $\{E,R\}$: there is no difficulty for the first sentence, and for the second, we make use of Corollary~\ref{cor:fo-def} and relativise to parts of $\Pp$.
Moreover, since dicings project to induced subgraphs as remarked above, and since $\Cc'$ is hereditary, the sentence $\psi$ can be chosen to be hereditary (if it holds on a structure then it holds on induced substructures), and in particular, in the vocabulary of Section~\ref{sec:stability}, $\psi$ induces itself on semi-elementary substructures.

Define $\eta(\tup x_1,\tup x_2; \tup y)$, where $\tup x_1$ and $\tup x_2$ are both tuples of variables of size 2, to be the formula given by
\[
    \eta(\tup x_1, \tup x_2; \tup y) = \phi(\tup x_1; \tup y) \wedge \lambda(\tup x_2; \tup y).
\]
Let $T$ be the theory of $\Cc$.
What we have proved so far can be stated as 
\[
    T \quad\textrm{implies}\quad \exists \tup s. \psi[R(\tup x_1, \tup x_2) / \eta(\tup x_1, \tup x_2; \tup s)].
\]
Therefore, by the Finitary Substitute Lemma (Lemma~\ref{lem:finitary-substitute}) there is a finitary formula $\eta'(\tup x_1, \tup x_2; \tup y)$ such that 
\[
    T\quad\textrm{implies}\quad \exists \tup s.\psi[R(\tup x_1, \tup x_2) / \eta'(\tup x_1, \tup x_2; \tup s)].
\]
Then the formulas
\[
    \phi'(\tup x; \tup y) = \exists z. \eta'(\tup x, z, z; \tup y) \quad \text{and} \quad \lambda'(\tup x; \tup y) = \exists z. \eta'(z, z ,\tup x; \tup y)
\]
yield the wanted result.
\end{proof}
\fi

\subsection{Canonical reduction to bounded treedepth}\label{subsec:sd-construction}

With \Cref{thm:shrub-inductive} in hand, we proceed to the proof of~\Cref{thm:canonization}.
Fix a class $\Cc$ of shrubdepth at most $d$.

\subparagraph*{Properties of the construction.} We describe a construction that, given a graph $G\in \Cc$, constructs a structure $\Aa(G)$ of the following shape.
\begin{itemize}
 \item $\Aa(G)$ is a structure over a signature consisting of several unary relations and one binary relation.
 Thus, we see $\Aa(G)$ as a vertex-colored directed graph, and we will apply the usual directed graphs terminology to $\Aa(G)$.
 \iffull For instance, the elements of $\Aa(G)$ are called {\em{vertices}} and pairs belonging to the binary relation are called {\em{arcs}}. \else \fi
 \item The vertex set of $G$ is contained in the vertex set of $\Aa(G)$. 
 The elements of $V(G)$ will be called {\em{leaves}} of $\Aa(G)$.
 In $\Aa(G)$ there is a unary predicate marking all the leaves.
 \item In $\Aa(G)$ there is a specified vertex, called the {\em{root}}, such that for every vertex $u$ of $\Aa(G)$ there is an arc from $u$ to the root. The root is identified using a unary predicate.
\end{itemize}
The construction will satisfy the following properties.
\begin{itemize}
 \item The mapping $G\mapsto \Aa(G)$ is isomorphism-invariant within the class $\Cc$.
 \item For every vertex $u$ of $\Aa(G)$, there are at most $c$ arcs with tail at $u$, for some constant $c$ depending only on $\Cc$.
 \item There is a transduction $\Ts$ depending on $\Cc$ such that $\Aa(G)\in \Ts(G)$.
 \item The class $\{\Aa(G) \mid G \in \Cc\}$ has bounded treedepth.
 \item There is an interpretation $\Is$ depending on $\Cc$ such that $G=\Is(\Aa(G))$. 
 \item Given $G$, $\Aa(G)$ can be computed in time $\Oh(n^2)$, where $n$ is the vertex count of $G$.
 \end{itemize}
We proceed by induction on $d$, the shrubdepth of $\Cc$\iffull.

\subparagraph*{Base case.} When $d=1$, all graphs in the class $\Cc$ have at most one vertex. Hence, for $G\in \Cc$, one can just output $\Aa(G)$ to be $G$ and declare the only vertex $u$ of $G$ to be the root of $\Aa(G)$.
(This involves also adding an arc $(u,u)$ and marking $u$ using unary predicates for the root and for the leaves.) Thus we can set $c=1$ and the definitions of $\Ts$ and $\Is$ are obvious.
\fi

\subparagraph*{Preparation for the inductive construction.} Suppose $d>1$. Let $\varphi(\tup x;\tup y)$, $\lambda(\tup x;\tup y)$, $\ell\in \N$, and $\Cc'$ be the finitary formulas, the bound, and the class provided by \Cref{thm:shrub-inductive}.
Since the shrubdepth of $\Cc'$ is at most $d-1$, by induction assumption we get a suitable mapping $\Aa'(\cdot)$, constant $c'$, transduction $\Ts'$, and interpretation $\Is'$ that satisfy the properties stated above for $\Cc'$.

Call a tuple $\tup s\in G^{\tup y}$ {\em{good}} if $(\varphi(G^{\tup x};\tup s),\lambda(G^{\tup x};\tup s))$ is a dicing of $G$ of order at most $\ell$ satisfying that for every part $A$ of $\varphi(G^{\tup x},\tup s)$ it holds that $G[A]\in \Cc'$.
Define
$$\Ff=\{\,(\varphi(G^{\tup x};\tup s),\lambda(G^{\tup x};\tup s))\ \colon\ \tup s\in G^{\tup y}\textrm{ is a good tuple}\,\}.$$
By Theorem~\ref{thm:shrub-inductive}, we have
$$1\leq |\Ff|\leq k$$
for some constant $k\in \N$ depending only on $\Cc$.

Let $\wh{\Ll}$ be the coarsest partition that refines all label partitions of the dicings belonging to~$\Ff$; that is, $u,v$ are in the same part of $\wh{\Ll}$ if and only if $u,v$ are in the same part of $\Ll$ for each $(\Pp,\Ll)\in \Ff$. Similarly, let $\wh{\Pp}$ be the coarsest partition that refines all component partitions of the dicings belonging to $\Ff$. Since $|\Ff|\leq k$ and $|\Ll|\leq \ell$ for each label partition featured in $\Ff$, we have
$$|\wh{\Ll}|\leq \ell^k.$$
Moreover, every part of $\wh{\Pp}$ is contained in a single part of any component partition featured in~$\Ff$, hence $G[B]\in \Cc'$ for every part $B$ of $\wh{\Pp}$.

Let $\wh{\Ff}=\{(\Pp,\wh{\Ll})\colon (\Pp,\Ll)\in \Ff\}$. Since $\wh{\Ll}$ refines each label partition featured in $\Ff$, it follows that every element of $\wh{\Ff}$ is a dicing of $G$. Then, for a component partition $\Pp$ featured in $\Ff$, let $Z_\Pp\subseteq \wh{\Ll}\times \wh{\Ll}$ be the symmetric relation witnessing that $(\Pp,\wh{\Ll})$ is a dicing.

\subparagraph*{Definition of the construction.}

We now describe the structure $\Aa(G)$; see Figure~\ref{fig:construction}.
Construct:
\begin{itemize}
 \item a root vertex $r$;
 \item for every part $L\in \wh{\Ll}$, a vertex $x_L$;
 \item for every component partition $\Pp$ featured in $\Ff$, a vertex $y_{\Pp}$;
 \item for every component partition $\Pp$ featured in $\Ff$, and every part $A\in \Pp$, a vertex $z_{\Pp,A}$; and
 \item for every component partition $\Pp$ featured in $\Ff$, and every (unordered) pair $LL'\in Z_\Pp$, a vertex $q_{\Pp,LL'}$. (Note here that $Z_\Pp$ is symmetric, so we may treat its elements as unordered pairs of elements of $\wh{\Ll}$.)
\end{itemize}
Further, for every part $B$ of $\wh{\Pp}$ we have $G[B]\in \Cc'$, hence we may apply the construction $\Aa'$ to $G[B]$, yielding a structure $H_B=\Aa'(G[B])$. Let $r_B$ be the root of $H_B$. We add all structures $H_B$ obtained in this way to $\Aa(G)$. We then connect these with the following arcs:
\begin{enumerate}
 \item for every vertex $u$ of $\Aa(G)$ there is an arc $(u,r)$;
 \item for every vertex of the form $z_{\Pp,A}$ there is an arc $(z_{\Pp,A},y_{\Pp})$;
 \item for every vertex of the form $q_{\Pp,LL'}$, there are arcs $(q_{\Pp,LL'},x_{L})$, $(q_{\Pp,LL'},x_{L'})$, and $(q_{\Pp,LL'},y_{\Pp})$;
 \item for every part $B$ of $\wh{\Pp}$ and every component partition $\Pp$ featured in $\Ff$, there is an arc $(r_B,z_{\Pp,A})$, where $A$ is the unique part of $\Pp$ that contains $B$;
 \item for every vertex $u$ of $G$ there is an arc $(u,x_L)$, where $L$ is the unique part of $\wh{\Ll}$ that contains~$u$. (Recall that the vertex set of $G$ is the union of the leaf sets of $H_B$ for $B\in \wh{\Pp}$.)

\end{enumerate}
\begin{figure}[t]
    \begin{center}
    \def\svgwidth{0.8\textwidth}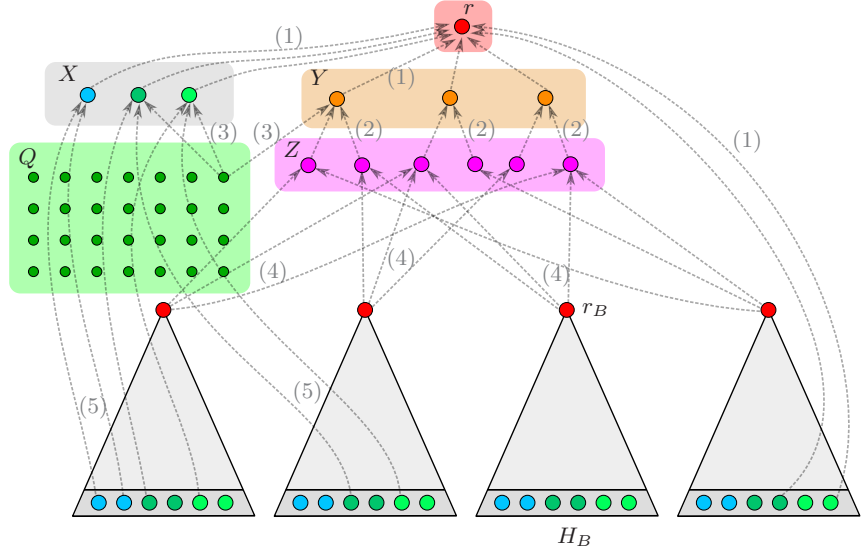
    \end{center}
    \caption{Inductive construction of $\Aa(G)$. Vertices of the form $r$, $y_{\Pp}$, $z_{\Pp,A}$, $q_{\Pp,LL'}$ are depicted in red, orange, violet, and green, respectively. Vertices of the form $x_L$ are depicted in the top-left corner of the figure in different soft colors (which do not correspond to unary predicates), matching the colors of vertices of $G$ that point to them; thus the soft color partition is $\wh{\Ll}$. We depict a few representatives for each type of arcs.} \label{fig:construction}
\end{figure}

Finally, we add five fresh unary predicates, called $R,X,Y,Z,Q$, respectively selecting the root~$r$, the vertices of the form $x_L$, the vertices of the form $y_{\Pp}$, the vertices of the form $z_{\Pp,A}$, and the vertices of the form $q_{\Pp,LL'}$. Note here that $H$ contains more unary predicates: those that come with structures $H_B$ constructed by induction. These include a unary relation selecting the leaves.

This concludes the construction of $\Aa(G)$. \iffull
 We now verify the required properties.

\subparagraph*{First properties.} The first two properties follow easily from the construction.

\begin{claim}\label{cl:iso-inv}
    The mapping $G\mapsto \Aa(G)$ is isomorphism-invariant.
\end{claim}

\begin{claim}\label{cl:outdeg}
    There is a constant $c$ depending only on $\Cc$ such that every vertex of $\Aa(G)$ is a tail of at most $c$ arcs.
\end{claim}
\begin{claimproof}
    Claim~\ref{cl:iso-inv} is immediate from the construction.
    For Claim~\ref{cl:outdeg}, vertices originating from the $H_B$'s are the tail of at most $c'+2+|\Ff|\leq c'+2+k$ arcs, while those of the form $r$, $x_L$, $y_\Pp$, $z_{\Pp,A}$, or $q_{\Pp,LL'}$ are the tail of at most $4$ arcs.
\end{claimproof}

\subparagraph*{Transducing $\Aa(G)$ from $G$.} We now show that $\Aa(G)$ can be transduced from $G$.

\begin{claim}\label{cl:trans}
    There is a first-order transduction $\Ts$ such that $\Aa(G)\in \Ts(G)$ for all $G\in \Cc$.
\end{claim}
\begin{claimproof}
    The transduction first uses $k'|y|$ colors, where $k' \leq k$, to mark tuples $\tup s_1,\ldots,\tup s_{k'}\in G^{\tup y}$ such that:
    \begin{itemize}
     \item every tuple $\tup s_i$ is good;
     \item the dicings $(\Pp_i,\Ll_i)=(\varphi(G^{\tup x},\tup s_i),\lambda(G^{\tup x},\tup s_i))$ for $i\in \{1,\ldots,k'\}$ are pairwise different;
     \item for every good tuple $\tup s$, there is $i$ such that $(\varphi(G^{\tup x},\tup s),\lambda(G^{\tup x},\tup s))=(\Pp_i,\Ll_i)$.
    \end{itemize}
    By \Cref{cor:fo-def}, that a tuple is good can be expressed by a first-order formula.
    Therefore, since we know that $\tup s_1,\dots,\tup s_{k'}$ as above exist, we may just guess them non-deterministically and verify, in first-order logic, that they satisfy the required properties. 
    Thus we constructed $\Ff=\{(\Pp_i,\Ll_i)\colon i\in \{1,\ldots,k'\}\}$.

    Once $\Ff$ is constructed, it is straightforward to interpret the partitions $\wh{\Pp}$ and $\wh{\Ll}$, and then construct (by copying suitable vertices of $G$) all the vertices of the form $r$, $x_L$, $y_\Pp$, $z_{\Pp,A}$, and $q_{\Pp,LL'}$, as well as interpret arcs between them. Next we apply transduction $\Ts'$ to $G[B]$ for every $B\in \wh{\Pp}$ in parallel. (This is done by suitably relativizing formulas featured in $\Ts'$ to parts of $\wh{\Pp}$.) This yields structures $H_B$ for all $B\in \wh{\Pp}$. Finally, the remaining arcs --- of the form $(r_B,z_{\Pp,A})$, $(u,r)$, or $(u,x_L)$, where $u$ is a vertex of $G$ --- are easy to interpret.
\end{claimproof}

Note that from Claim~\ref{cl:trans} we immediately obtain the following.

\begin{claim}
 If $G\in \Cc$ has $n$ vertices, then $\Aa(G)$ has $\Oh(n)$ elements.
\end{claim}

\subparagraph*{Bounding the treedepth of $\Aa(G)$.} A bound on the treedepth easily follows.

\begin{claim}\label{cl:td-bnd}
    The class $\{\Aa(G)\colon G\in \Cc\}$ has bounded treedepth.
\end{claim}
\begin{proof}
    Let $\Dd$ be the class in question.
    It follows from~\Cref{thm:sd-closed-under-transductions} and Claim~\ref{cl:trans} that $\Dd$ has bounded shrubdepth. Consider now any $H\in \Dd$. By Claim~\ref{cl:outdeg}, every subgraph of the Gaifman graph of $H$, say on $n$ vertices, has at most $cn$ edges, hence it contains a vertex of degree at most~$2c$. It follows that $\Dd$ is $2c$-degenerate, and it is well-known\footnote{This can be argued as follows: A class of bounded shrubdepth excludes some path as an induced subgraph~\cite{GanianHNOM19}, which together with the assumption that the class is degenerate, implies that the class has bounded treedepth~\cite[Proposition~6.4]{sparsity}.} that degenerate classes of bounded shrubdepth in fact have bounded treedepth.
\end{proof}

\subparagraph*{Interpreting $G$ in $\Aa(G)$.} We derive from Claims~\ref{cl:outdeg} and~\ref{cl:trans} that the following holds.

\begin{claim}\label{cl:inter}
    There is an interpretation $\Is$ such that $G=\Is(\Aa(G))$. 
    \end{claim}
    \begin{claimproof}
    First, the domain formula $\gamma_\Is(u)$ just selects all the leaves of $\Aa(G)$; recall that these are marked using a unary predicate.
     
     The adjacency formula $\alpha_\Is(u,v)$ works as follows. Given vertices $u$ and $v$, we first verify that they are distinct leaves. Then, we check whether there is a part $B\in \wh{\Pp}$ such that $u,v\in B$. This can be done by checking whether there exists a vertex of the form $r_B$ for $B\in \wh{\Pp}$ such that both arcs $(u,r_B)$ and $(v,r_B)$ are present. (Recall here that by construction, all vertices of the form $r_B$ for $B\in \wh{\Pp}$ are marked using a unary predicate.) 
     
     If this is the case, then whether the edge $uv$ is present in $G$ can be inferred using the interpretation $\Is'$ obtained by induction for the class $\Cc'$. This interpretation is applied to the structure $H_B$ by suitably relativizing all formulas involved to vertices of $H_B$, that is, vertices~$w$ for which the arc $(w,r_B)$ is present.
     
     If this is not the case, that is, $u\in B$ and $v\in B'$ where $B$ and $B'$ are distinct parts of~$\wh{\Pp}$, then adjacency between $u$ and $v$ in $G$ is equivalent to the following assertion: if $L,L'$ are the parts of $\wh{\Ll}$ to which $u$ and $v$ belong, respectively, and $\Pp$ is any component partition featured in $\Ff$ such that $u$ and $v$ belong to different parts of $\Pp$ (note that such $\Pp$ must exist by the definition of $\wh{\Pp}$), then the vertex $q_{\Pp,LL'}$ exists. This is equivalent to the following first-order formula: there exist vertices $r_B,r_{B'}$, vertices $x,x'$ satisfying $X$ (possibly equal), distinct vertices $z,z'$ satisfying $Z$, a vertex $y$ satisfying $Y$, and a vertex $q$ satisfying $Q$, such that arcs 
     $$(u,x),(v,x'),(u,r_B),(v,r_{B'}),(r_B,z),(r_{B'},z'),(z,y),(z',y),(q,y),(q,x),(q,x')$$
     are all present, and $\{(q,z),(q,x),(q,x')\}$ are all the arcs with tail $q$. This concludes the definition of $\alpha_\Is$.
\end{claimproof}

\subparagraph*{Construction of $\Aa(G)$ in $\Oh(n^2)$ time.} To prove the claimed upper bound, we will invoke the following algorithmic meta-theorem.

\begin{restatable}[Obtained by combining~\cite{FominK22} and~\cite{KazanaS13}]{lemma}{cwmetathm}\label{lem:meta-thm}
    Let $\Cc$ be a class of colored graphs of bounded cliquewidth and $\phi(\tup x)$ be a formula.
    For every $G \in \Cc$ of vertex count $n$ one may, after an $\Oh(n^2)$ preprocessing,
    \begin{enumerate}[(i)]
        \item enumerate $\varphi(G^{\tup x})$ with constant delay;
        \item given a tuple $\tup u \in G^{\tup x}$, determine whether $G \models \phi(\tup u)$ in constant time.
    \end{enumerate}
\end{restatable}

We include a proof in Appendix~\ref{app:cw-metathm}.
The wanted result follows from the lemma.

\begin{claim}\label{cl:algo}
    Given $G\in \Cc$ with $n$ vertices, $\Aa(G)$ can be constructed in time $\Oh(n^2)$. 
\end{claim}

\begin{claimproof}
    As we have already observed, a tuple $\tup s$ being good can be expressed by a first-order formula.
    Therefore, we may obtain in time $\Oh(n^2)$ a good tuple $\tup s_1$.
    To obtain $\tup s_2$, one should use colors to mark $\tup s_1$, and then alter the formula to additionally require that the partition defined by $\tup s_2$ is different from the one defined by $\tup s_1$.
    Proceeding iteratively, since $k' \leq k$ is a constant, one obtains tuples $\tup s_1,\dots, \tup s_{k'}$ as in the proof of Claim~\ref{cl:trans} in time $\Oh(n^2)$.

    Then we have that $\Ff=\{(\varphi(G^{\tup x},\tup s_i),\lambda(G^{\tup x},\tup s_i)) \colon i \in \{1,\ldots,k'\}\}$, and each of the $k'$ different $(\varphi(G^{\tup x},\tup s_i),\lambda(G^{\tup x},\tup s_i))$ can be computed in $\Oh(n^2)$ by making $2n^2$ queries as in the second item of~\Cref{lem:meta-thm}.
    Once $\Ff$ is computed, partitions $\wh{\Pp}$ and $\wh{\Ll}$ can be computed in time $\Oh(n^2)$. Then vertices of the form $r$, $x_L$, $y_{\Pp}$, $z_{\Pp,A}$, and $q_{\Pp,LL'}$, as well as arcs between them, can be constructed in time $\Oh(n^2)$ by just following the definitions. Next, for every part $B\in \wh{\Pp}$ we construct $H_B=\Aa'(G[B])$ by applying the algorithm for the class $\Cc'$ provided by induction.
    This takes time $\Oh(|B|^2)$ per part $B\in \wh{\Pp}$, so time $\sum_{B\in \wh{\Pp}} \Oh(|B|^2)\leq \Oh(n^2)$ in total. Finally, arcs connecting structures $H_B$ with the other vertices can be easily constructed in time $\Oh(n^2)$. 
\end{claimproof}
   
From Claims~\ref{cl:iso-inv} and~\ref{cl:inter} it follows that for graphs $G,G'\in \Cc$, structures $\Aa(G)$ and $\Aa'(G)$ are isomorphic if and only if $G$ and $G'$ are. Together with the algorithmic statement provided by Claim~\ref{cl:algo} and the treedepth bound provided by Claim~\ref{cl:td-bnd}, we conclude with \Cref{thm:canonization}.
\fi

\subsection{Canonization and isomorphism test}\label{subsec:sd-isomorphism}

We now use \Cref{thm:canonization} to prove \Cref{thm:gi}, that is, give a quadratic-time isomorphism test for any class of graphs of bounded shrubdepth.
As mentioned in the introduction, in fact we solve the more general canonization problem, defined as follows.

For a class of structures $\Cc$, a {\em{canonization map}} for $\Cc$ is a mapping $\cn$ with the following property: for every $\str M\in \Cc$, $\cn(\str M)$ is a total order on elements of $\str M$ so that if $\str M,\str M'\in \Cc$ are isomorphic, then associating elements with same index in $\cn(\str M)$ and in $\cn(\str M')$ yields an isomorphism between $\str M$ and $\str M'$.
Note that if there is a canonization map $\cn$ for $\Cc$ that is efficiently computable, then this immediately gives an isomorphism test within the same time complexity.

For classes of bounded treedepth, Bouland et al.~\cite{BoulandDK12} gave a relatively easy fixed-parameter isomorphism test. Their techniques can be easily extended to the canonization problem for binary structures of bounded treedepth.

\begin{restatable}[Adapted from~\cite{BoulandDK12}]{theorem}{tdcanonization}\label{thm:td-canonization}
 Let $\Dd$ be a class of binary structures of bounded treedepth. There exists a canonization map $\cn$ on $\Dd$ that is computable in time $\Oh(n\log^2 n)$, where $n$ is the size of the universe of the input structure.
\end{restatable}

\iffull
A proof is included in Appendix~\ref{app:td-canonization}.  \else\fi
We can now prove the main result of this section.

\begin{theorem}\label{thm:shb-canonization}
 Let $\Cc$ be a class of graphs of bounded shrubdepth. There exists a canonization map $\cn$ on $\Cc$ that is computable in time $\Oh(n^2)$, where $n$ is the vertex count of the input graph.
\end{theorem}
\begin{proof}
 Let $\DD$ be the class of bounded treedepth and $\Aa\colon \Cc\to \DD$ be the mapping provided by \Cref{thm:canonization} for the class $\Cc$. Then to get a suitable canonization map for $\Cc$, it suffices to compose $\Aa$ with the canonization map for $\DD$, provided by \Cref{thm:td-canonization}, and restrict the output order to the vertex set of the original graph. 
\end{proof}

As discussed, \Cref{thm:gi} follows immediately from \Cref{thm:shb-canonization}.

%% file: construction.pdf_tex
\begingroup%
  \makeatletter%
  \providecommand\color[2][]{%
    \errmessage{(Inkscape) Color is used for the text in Inkscape, but the package 'color.sty' is not loaded}%
    \renewcommand\color[2][]{}%
  }%
  \providecommand\transparent[1]{%
    \errmessage{(Inkscape) Transparency is used (non-zero) for the text in Inkscape, but the package 'transparent.sty' is not loaded}%
    \renewcommand\transparent[1]{}%
  }%
  \providecommand\rotatebox[2]{#2}%
  \newcommand*\fsize{\dimexpr\f@size pt\relax}%
  \newcommand*\lineheight[1]{\fontsize{\fsize}{#1\fsize}\selectfont}%
  \ifx\svgwidth\undefined%
    \setlength{\unitlength}{371.43684799bp}%
    \ifx\svgscale\undefined%
      \relax%
    \else%
      \setlength{\unitlength}{\unitlength * \real{\svgscale}}%
    \fi%
  \else%
    \setlength{\unitlength}{\svgwidth}%
  \fi%
  \global\let\svgwidth\undefined%
  \global\let\svgscale\undefined%
  \makeatother%
  \begin{picture}(1,0.64564914)%
    \lineheight{1}%
    \setlength\tabcolsep{0pt}%
    \put(0,0){\includegraphics[width=\unitlength,page=1]{construction.pdf}}%
    \put(0.67309626,0.27572326){\color[rgb]{0,0,0}\makebox(0,0)[lt]{\lineheight{1.25}\smash{\begin{tabular}[t]{l}{\tn{$r_B$}}\end{tabular}}}}%
    \put(0,0){\includegraphics[width=\unitlength,page=2]{construction.pdf}}%
    \put(0.53369129,0.62937832){\color[rgb]{0,0,0}\makebox(0,0)[lt]{\lineheight{1.25}\smash{\begin{tabular}[t]{l}{\tn{$r$}}\end{tabular}}}}%
    \put(0,0){\includegraphics[width=\unitlength,page=3]{construction.pdf}}%
    \put(0.05755461,0.55069628){\color[rgb]{0,0,0}\makebox(0,0)[lt]{\lineheight{1.25}\smash{\begin{tabular}[t]{l}{\tn{$X$}}\end{tabular}}}}%
    \put(0,0){\includegraphics[width=\unitlength,page=4]{construction.pdf}}%
    \put(0.35350402,0.54563828){\color[rgb]{0,0,0}\makebox(0,0)[lt]{\lineheight{1.25}\smash{\begin{tabular}[t]{l}{\tn{$Y$}}\end{tabular}}}}%
    \put(0,0){\includegraphics[width=\unitlength,page=5]{construction.pdf}}%
    \put(0.32214648,0.46069953){\color[rgb]{0,0,0}\makebox(0,0)[lt]{\lineheight{1.25}\smash{\begin{tabular}[t]{l}{\tn{$Z$}}\end{tabular}}}}%
    \put(0.29258504,0.3173469){\color[rgb]{0.4,0.4,0.4}\makebox(0,0)[lt]{\lineheight{1.25}\smash{\begin{tabular}[t]{l}{\tn{$(4)$}}\end{tabular}}}}%
    \put(0.08152268,0.15758557){\color[rgb]{0.4,0.4,0.4}\makebox(0,0)[lt]{\lineheight{1.25}\smash{\begin{tabular}[t]{l}{\tn{$(5)$}}\end{tabular}}}}%
    \put(0,0){\includegraphics[width=\unitlength,page=6]{construction.pdf}}%
    \put(0.64358263,0.00336282){\color[rgb]{0,0,0}\makebox(0,0)[lt]{\lineheight{1.25}\smash{\begin{tabular}[t]{l}{\tn{$H_B$}}\end{tabular}}}}%
    \put(0,0){\includegraphics[width=\unitlength,page=7]{construction.pdf}}%
    \put(0.01079783,0.45247984){\color[rgb]{0,0,0}\makebox(0,0)[lt]{\lineheight{1.25}\smash{\begin{tabular}[t]{l}{\tn{$Q$}}\end{tabular}}}}%
    \put(0.84952563,0.47211236){\color[rgb]{0.4,0.4,0.4}\makebox(0,0)[lt]{\lineheight{1.25}\smash{\begin{tabular}[t]{l}{\tn{$(1)$}}\end{tabular}}}}%
    \put(0.6485856,0.48512082){\color[rgb]{0.4,0.4,0.4}\makebox(0,0)[lt]{\lineheight{1.25}\smash{\begin{tabular}[t]{l}{\tn{$(2)$}}\end{tabular}}}}%
    \put(0.23464167,0.48131942){\color[rgb]{0.4,0.4,0.4}\makebox(0,0)[lt]{\lineheight{1.25}\smash{\begin{tabular}[t]{l}{\tn{$(3)$}}\end{tabular}}}}%
    \put(0,0){\includegraphics[width=\unitlength,page=8]{construction.pdf}}%
    \put(0.28523953,0.48185543){\color[rgb]{0.4,0.4,0.4}\makebox(0,0)[lt]{\lineheight{1.25}\smash{\begin{tabular}[t]{l}{\tn{$(3)$}}\end{tabular}}}}%
    \put(0.44312197,0.33245679){\color[rgb]{0.4,0.4,0.4}\makebox(0,0)[lt]{\lineheight{1.25}\smash{\begin{tabular}[t]{l}{\tn{$(4)$}}\end{tabular}}}}%
    \put(0.62545905,0.31389618){\color[rgb]{0.4,0.4,0.4}\makebox(0,0)[lt]{\lineheight{1.25}\smash{\begin{tabular}[t]{l}{\tn{$(4)$}}\end{tabular}}}}%
    \put(0.33297724,0.17483019){\color[rgb]{0.4,0.4,0.4}\makebox(0,0)[lt]{\lineheight{1.25}\smash{\begin{tabular}[t]{l}{\tn{$(5)$}}\end{tabular}}}}%
    \put(0.31133185,0.59530496){\color[rgb]{0.4,0.4,0.4}\makebox(0,0)[lt]{\lineheight{1.25}\smash{\begin{tabular}[t]{l}{\tn{$(1)$}}\end{tabular}}}}%
    \put(0.44309052,0.54733069){\color[rgb]{0.4,0.4,0.4}\makebox(0,0)[lt]{\lineheight{1.25}\smash{\begin{tabular}[t]{l}{\tn{$(1)$}}\end{tabular}}}}%
    \put(0.53752309,0.48512082){\color[rgb]{0.4,0.4,0.4}\makebox(0,0)[lt]{\lineheight{1.25}\smash{\begin{tabular}[t]{l}{\tn{$(2)$}}\end{tabular}}}}%
    \put(0.40526404,0.48512082){\color[rgb]{0.4,0.4,0.4}\makebox(0,0)[lt]{\lineheight{1.25}\smash{\begin{tabular}[t]{l}{\tn{$(2)$}}\end{tabular}}}}%
  \end{picture}%
\endgroup%

%% file: definable-strategies.tex
\section{Canonical strategies in monadically stable classes}\label{sec:definable_strats}

In this section, we give canonical strategies in (a variant of) the Flipper game, whose termination in a bounded number of rounds characterizes monadic stability~\cite{flipper_game}.
Stating our main result requires a few definitions which are given now.

\subsection{Definitions and formal statement}

In this section, all structures are defined over the signature of graphs, possibly with additional constants and binary relations (which will always define equivalence classes of bounded index).

\paragraph*{Flips.} Let $\Pp$ be a partition of a structure $\str M$.
A $\Pp$-flip $\str M'$ of $\str M$ is any structure with the same domain, constants and equivalence relations as $\str M$, and such that there is a symmetric relation $Z \subseteq \Pp \times \Pp$ satisfying for all $u,v \in \str M$,
\[
    \textrm{$u,v$ are adjacent in $\str M'$}\quad\textrm{if and only if}\quad \textrm{($u,v$ are adjacent in $\str M$)}\ \mathsf{xor}\ \textrm{(}(\Pp(u),\Pp(v))\in Z\textrm{)},
\]
where $\Pp(x)$ denotes the part of $\Pp$ to which $x$ belongs.
Note that for a given part $A \in \Pp$, it may or may not be that $(A,A) \in Z$, giving rise to different flips.

Given a model $\str M$, a partition $\Pp$ of $\str M$, an integer $r \in \N$, and vertices $u,v \in \str M$, we say that $u$ and $v$ are $r$-independent over $\Pp$ (in $\str M$), written $u \ind[\Pp]r v$, if there is a $\Pp$-flip of $\str M$ in which there is no path of length $\leq r$ from $u$ to $v$.
Otherwise (in each $\Pp$-flip, there is a path of length $\leq r$ from $u$ to $v$), we write $u \nind[\Pp]r v$, and we let
\[
    B^r_{\Pp,\str M}(v) = \{u \in \str M \mid u \nind[\Pp]r v\}
\]
denote the corresponding ball around $v$.
The following property is easy to check.

\begin{lemma}\label{lem:flip_balls_and_substructures}
Let $\str M$ be an induced substructure of a structure $\str N$, let $\Pp$ be a partition of $\str M$ and let $\Qq$ be a partition of $\str M$ which is finer than the restriction of $\Pp$ to $\str M$.
Then for all $v \in \str M$ and all $r \in \N$, $B^r_{\Qq,\str M}(v) \subseteq B^r_{\Pp,\str N}(v)$.
\end{lemma}

Given two partitions $\Pp,\Pp'$ of a set $X$, we let $\Pp \refine \Pp'$ denote their coarsest common refinement, meaning the partition
\[
    \Pp \refine \Pp' = \{P \cap P' \mid P \in \Pp \text{ and } P' \in \Pp'\} \setminus \{\emptyset\}.
\]
Note that if $\Pp$ and $\Pp'$ have index $\leq k$ and $\leq k'$, then $\Pp \refine \Pp'$ has index $\leq kk'$.

\paragraph*{Separation game.} The \emph{radius-$r$ separation\footnote{In~\cite{flipper_game}, this variant of the game (with index $2$) is called the ``Pseudo-Flipper game''.} game of index $k$} over a structure $\str M$ proceeds in rounds $i=1,2,\dots$ as follows.
Initially, the arena $A_1$ is the whole structure.
In the $i$-th round,
\begin{itemize}
\item Separator plays a partition $\Pp_i$ of index at most $k$;
\item Connector plays a vertex $c_i \in A_i$;
\item We set
\[
    A_{i+1} = A_i \cap B^r_{\Pp_1 \refine \dots \refine \Pp_i}(c_i),
\]
and if $|A_{i+1}|=1$, the game stops and Separator wins.
\end{itemize}
We sometimes abbreviate it as the $(r,k)$-separation game.
We also remark that unlike in the Splitter game, the ball $B^r_{\Pp_1 \refine \dots \refine \Pp_i}(c_i)$ is computed in $\str M$ and not in the current arena.

An important variant is the \emph{atomic} separation\footnote{In~\cite{flipper_game}, this variant (with tuple-size $1$) is called the ``Separation game''.} game of tuple-size $k$, where in each round $i$, Separator picks a tuple $\tup s_i$ of at most $k$ vertices and sets $\Pp_i$ to be the partition into atomic edge-types over $\tup s_i$ (whose index is $\leq 2^k + k$).

Note that for the separation game, the index $k$ is somewhat irrelevant: if Separator wins the $(r,k)$-game in at most $\ell$ rounds, then he wins the $(r,2)$-game in at most $k \ell$ rounds.
A similar remark holds for the atomic version; however the fixed finite index is important in our definable version (see below).
One of the main results of~\cite{flipper_game} is that the (atomic) separation game characterizes monadic stability.

\begin{theorem}[Theorem 1.4 in \cite{flipper_game}]
    Let $\Cc$ be a class of graphs. The following are equivalent:
    \begin{enumerate}
        \item $\Cc$ is monadically stable;
        \item for every $r$ there is $\ell$ such that on all graphs of $\Cc$, Separator wins the radius-$r$ separation game of index $2$ in $\ell$ rounds;
        \item for every $r$ there is $\ell$ such that on all graphs of $\Cc$, Separator wins the atomic radius-$r$ separation game with tuple size $1$ in $\ell$ rounds.
    \end{enumerate}
\end{theorem}

We may now state our main result in this section, which expresses the existence of canonical Separator-strategies in monadically stable classes.

\begin{theorem}[Canonical Separator-strategies]\label{thm:definable_strategies}
Let $\Cc$ be a monadically stable class and let $r \geq 1$ and $\ell$ be such that on all graphs of $\Cc$, Separator wins the atomic radius-$r$ separation game of tuple-size $1$ in $\ell$ rounds.
There are formulas $\delta_1(\tup x;y_1),\delta_2(\tup x; y_1, y_2),\dots, \delta_\ell(\tup x; y_1, \dots, y_\ell)$ defining partitions of index $\leq k$ for some integer $k$, such that on all graphs $\str M \in \Cc$, Separator wins the $(r,k)$-separation game in $\ell$ rounds by playing, in round $i$, the partition $\Pp_i=\delta_i(\str M^{\tup x}; c_1, \dots, c_i)$, where $c_1,\dots,c_i$ are the previous connector moves.
\end{theorem}

\subsection{Proof of Theorem~\ref{thm:definable_strategies}}

We fix a monadically stable class of graphs $\Cc$ and $r,\ell \geq 1$ such that on all graphs of $\Cc$, Separator wins the atomic radius-$r$ separation game of tuple size $1$ in $\ell$ rounds.
Let us observe that this can be expressed by a first-order sentence.
We let $T$ be the theory of $\Cc$ and $d$ be twice the ladder index of $T$ plus 1.

\paragraph*{Configurations.} A \emph{configuration after $i$ rounds} in the $r$-separation game over $\str M$ is a collection of $i$ partitions $\Pp_1,\dots,\Pp_i$ and $i$ connector moves $c_1,\dots,c_i$.
We write it $C_i=(\Pp_1,c_1,\Pp_2,c_2,\dots,\Pp_i,c_i)$.
Such a collection is \emph{winning} if either of the two following holds:
\begin{itemize}
\item for some $j \in \{1,\dots, i\}$, connector's move $c_j$ does not belong to the arena at time $j$,
\[
    A_j = \bigcap_{t =1}^{j-1} B_{\Pp_1 \refine \dots \refine \Pp_{t}}^r(c_t)
\]
\item $A_{i+1}$ is a singleton.
\end{itemize}
This can be expressed by a first-order sentence (over the signature $\{E, \Pp_1,\dots, \Pp_i,c_1,\dots,c_i\}$).
The $A_j$'s as defined above are called the \emph{arenas associated to the configuration $C_i$.}

We say that a configuration after $i$ rounds $C_i$ \emph{leads to a victory in $\ell$ rounds in the atomic game of tuple size $k$} if, starting from $C_i$, separator wins the atomic game of tuple size $k$ with $\ell -i$ remaining moves.
More formally, for all $j \in \{i+1,\dots, \ell\}$ and for all sequences $c_{i+1},\dots, c_{j}$ of additional connector moves, there is a tuple $t_{c_{i+1} \dots c_{j}}$ with corresponding partition $\Pp_{c_{i+1} \dots c_{j}}$, defined so that for any sequence $c_{i+1},\dots, c_{\ell}$, the configuration 
\[
    (c_1,\Pp_1,\dots,c_i,\Pp_i,c_{i+1},\Pp_{c_{i+1}},c_{i+2},\Pp_{c_{i+1}c_{i+2}},\dots, c_\ell,\Pp_{c_{i+1}\dots c_{\ell}})
\]
is winning.
Yet again, this may be expressed by a first-order sentence over the signature $\{E,\Pp_1,\dots,\Pp_i,c_1,\dots,c_i\}$.

\paragraph*{Induction hypothesis.} We will prove the theorem by induction, with the following hypothesis~$\Hh_i$:

\begin{center}
there are formulas $\delta_1(\tup x; y_1), \dots, \delta_i(\tup x; y_1,\dots,y_i)$ defining partitions of finite index such that $T$ implies the sentence
``for all connector moves $c_1,\dots, c_i$, the configuration after $i$ rounds $(c_1,\Pp_1,\dots,c_i,\Pp_i)$ leads to a victory in $\ell$ rounds in the atomic game of tuple size $(d)^i$, where for all $j \in \{1,\dots,i\}$, $\Pp_j$ is the partition defined by $\delta_j(\tup x; c_1, \dots, c_j)$.''
\end{center}
As required, $\Hh_\ell$ implies the theorem, with $k$ being the maximum index of $\delta_1,\dots,\delta_\ell$ (this depends only on $T$ since having index $\leq p$ can be expressed by a first-order formula).
Note also that $\Hh_0$ corresponds to our working assumption that Separator wins the atomic radius-$r$ game of tuple size $1$ in $\ell$ rounds; hence there remains to establish the inductive step.

\paragraph*{Inductive step.} Let $i \in \{1,\dots, \ell\}$, and assume that $\Hh_{i-1}$ holds, we aim to prove $\Hh_i$.
We let $\psi_i$ be a sentence over the signature $\{E \cup R\}$, where $R$ is a fresh symbol of arity $2$, expressing
\begin{center}
$R$ is a partition such that for all connector moves $c_1,\dots,c_{i-1},c_{i}$, the configuration after $i$ rounds $(c_1,\Pp_1,\dots,c_{i-1},\Pp_{i-1},c_i,R)$, where for $j \in \{1,\dots,i-1\}$, $\Pp_j$ is the partition defined by $\delta_j(\tup x,c_1, \dots, c_j)$, leads to a victory in $\ell$ rounds in the atomic game of tuple size $(d)^{i-1}$.
\end{center}
Let $\phi(\tup x; \tup y)$ be a sentence defining the partition into atomic $\tup y$-types, where $|\tup y|=(d)^{i-1}$.
Our inductive hypothesis $\Hh_{i-1}$ tells us that $T$ implies the sentence $\exists \tup s.\psi[R(\tup x) / \phi(\tup x; \tup s)]$.
We now define $\psi'_i$ to be exactly like $\psi_i$, except that the tuple size (at the very end) is increased to $d^i$.

\begin{claim}
In the theory $T$, $\psi_i$ induces $\psi'_i$ on semi-elementary substructures.
\end{claim}

\begin{claimproof}
Let $\str M$ be a model of $T$, $\str N$ be an elementary extension of $\str M$ and $\Rr$ a partition such that $\str N[R / \Rr] \models \psi_i$: for all connector moves $c_1,\dots,c_{i-1},c_i \in \str N$, the configuration after $i$ rounds $(c_1,\Pp_1,\dots,c_{i-1},\Pp_{i-1},c_i,\Rr)$, where for $j \in \{1,\dots,i-1\}$, $\Pp_j$ is the partition defined by $\delta_j(\tup x,c_1 \dots c_j)$, leads to a victory in $\ell$ rounds in the atomic game of tuple size $(d)^{i-1}$.
Unpacking the definition, we get, for each $j \in \{i+1,\dots,\ell\}$, and for all sequences $c_{i+1},\dots,c_{j} \in \str N$ of connector moves, a tuple $\tup t_{c_{i+1} \dots c_j}\in \str N^{d^{i-1}}$ with corresponding partition $\Pp_{c_{i+1} \dots c_j}$ of atomic types over $\str N$, such that for any sequence $c_{i+1},\dots, c_{\ell} \in \str N$, the configuration 
\[
    C_{c_{i+1} \dots c_\ell} =(c_1,\Pp_1,\dots,c_{i-1},\Pp_{i-1},c_i,\Rr,c_{i+1},\Pp_{c_{i+1}},c_{i+2},\Pp_{c_{i+1}c_{i+2}},\dots, c_\ell,\Pp_{c_{i+1}\dots c_{\ell}})
\]
is winning.
Using definability of types (\Cref{thm:definability-of-types}) we get, for each $j \in \{i+1,\dots, \ell\}$ and each sequence $c_{i+1},\dots, c_j$ of connector moves in $\str M$, a tuple $\tup t'_{c_{i+1} \dots c_j} \in \str M^{d^i}$ whose corresponding partition $\Qq_{c_{i+1} \dots c_j}$ into atomic types over $\str M$ is finer than the restriction of $\Pp_{c_{i+1} \dots c_j}$ to $\str M$.
Now for a given sequence $c_{i+1},\dots,c_\ell \in \str M$ of connector moves, Lemma~\ref{lem:flip_balls_and_substructures} tells us that the arenas $A'_1,A'_2,\dots, A'_\ell$ associated to the configuration 
\[
    C'_{c_{i+1} \dots c_\ell}=(c_1,\Pp_1|_{\str M},\dots,c_{i-1},\Pp_{i-1}|_{\str M},c_i,\Rr|_{\str M},c_{i+1},\Qq_{c_{i+1}},c_{i+2},\Qq_{c_{i+1}c_{i+2}},\dots, c_\ell,\Qq_{c_{i+1}\dots c_{\ell}})
\]
are included (for all $j$, $A'_j \subseteq A_j$) in the arenas $A_1,A_2, \dots, A_\ell$ associated to $C_{c_{i+1} \dots c_\ell}$.
Hence, $C'_{c_{i+1} \dots c_\ell}$ is winning.
Since for all $j \in \{1,\dots, i-1\}$ it holds by definition of elementary extensions that $\Pp_{i-1}|_{\str M}$ coincides with the partition defined by $\delta_i(\tup x; c_1, \dots, c_j)$ over $\str M$, we get that $\str M[R / \Rr|_{\str M}] \models \psi'$, proving the claim.
\end{claimproof}

Thanks to the claim, and since $T$ implies $\exists \tup s. \psi[R(\tup x) / \phi(\tup x; \tup s)]$ as observed above, we may apply the Finitary Substitute Lemma (Lemma~\ref{lem:finitary-substitute}) and we get a finitary formula $\phi'(\tup x; \tup z)$ and a tuple $\tup r \in \str M^{\tup z}$ such that $T \models \exists \tup s. \psi'[R(\tup x) / \phi'(\tup x; \tup s)]$ and the partitions defined on $\str M$ by $\phi(\tup x, \tup s)$ and $\phi'(\tup x, \tup r)$ coincide.
In particular, the index of the partition defined by $\phi'(\tup x, \tup r)$ is bounded by $2^{d^{i-1}} + d^{i-1}$.
We let $\phi''(\tup x; \tup z)$ be a formula expressing that ``$\phi'(\tup x; \tup z)$ holds or the partition $\phi'(\cdot; \tup z)$ has index $> 2^{(d)^{i-1}} + d^{i-1}$''.
Note that $\phi''$ is finitary, as for some tuples it defines the same partition as $\phi$ and for others it defines the full partition of index $1$.
We now let 
\[
    \delta_{i}(\tup x) = \forall \tup y. \phi''(\tup x; \tup y).
\]
Then $\delta_i(\str M^{\tup x})$ is a partition which is finer than each of the finitely many partitions $\phi''_i(\tup x; \tup z)$, all of which have index $\leq 2^{d^{i-1}} + d^{i-1}$.
It follows that $\delta_i(\str M^{\tup x})$ has finite index and since $T \models \exists \tup s. \psi'[R(\tup x)/\phi'(\tup x; \tup s)]$, we get that $\Hh_i$ holds.

%% file: appendix.tex
\section{Proof of Lemma~\ref{thm:incremental}}\label{app:label-partition}

We now prove Lemma~\ref{thm:incremental}, which we restate for convenience.

\labelpartition*

Before moving on to the proof, we recall the statement from \cite[Theorem 3.5]{boundedLocalCliquewidth}.
For this we need two definitions.
\begin{definition}[Definition 3.2 from~\cite{boundedLocalCliquewidth}]
   Let $E \subseteq A \times B$ be a binary relation.
   Say that $E$ has a \emph{duality of order $k$} if at least one of two cases holds:
   \begin{itemize}
       \item there is a set $A_0 \subseteq A$ of size at most $k$ such that for every $b \in B$ there is some $a \in A_0$ with $\neg E(a, b)$; or
       \item there is a set $B_0 \subseteq B$ of size at most $k$ such that for every $a \in A$ there is some $b \in B_0$ with $E(a, b)$.
   \end{itemize}
\end{definition}
\begin{definition}
   Fix a set $V$.
   A symmetric function $f : V \times V \to \R_{\ge 0} \cup \set{+\infty}$ is called a \emph{pseudometric} if it satisfies the triangle inequality.
\end{definition}

\begin{theorem}[Theorem 3.5 of~\cite{boundedLocalCliquewidth}]
   \label{thm:original_incremental}
   Fix $r, k, t \in N$.
   Let $V$ be a finite set equipped with:
   \begin{itemize}
       \item a binary relation $E \subseteq V \times V$ such that for all $A \subseteq V$ and $B \subseteq V$, $E \cap (A \times B)$ has a duality of order $k$,
       \item a pseudometric dist: $V \times V \to \R_{\ge 0} \cup \set{+\infty}$,
       \item a partition $\Ll$ of $V$ with $|\Ll| \le t$,
   \end{itemize}
   such that $E(u, v)$ depends only on parts of $u$ and $v$ in $\Ll$ for all $u, v$ with $\dist(u, v) > r$.
   Then there is a set $S \subseteq V$ of size $\Oo(kt^2)$ such that $E(u, v)$ depends only on $E(u, S)$ and $E(S, v)$ for all $u, v \in V$ with $\dist(u, v) > 5r$.    
\end{theorem}

\begin{proof}[Proof of \Cref{thm:incremental}]
   Fix any $G \in \Cc$ and a dicing $(\Pp, \Ll)$ of $G$ of order at most $t$.
   Observe that as $\Cc$ is monadically stable (so in particular it has bounded VC dimension), by \Cref{thm:original_incremental} there is a constant $k$ depending only on $\Cc$ such that $E(G)$ has a duality of order $k$.
   We can define a pseudometric $f$ on $V(G)$ as follows:
   \[
       \dist(u, v) =
       \begin{cases}
           0 &\text{ if $u, v$ in the same part of $\Pp$;}\\
           +\infty &\text{ otherwise.}
       \end{cases}
   \]
   By \Cref{thm:original_incremental} we get that there is a set $S \subseteq V(G)$ of size $\Oo(t^2)$ (where the constant hidden in the $\Oo$ notation depends only on the class $\Cc$) such that for any $u, v \in V(G)$ if they are in different parts of $\Pp$, then $E(u, v)$ depends only on the parts of $u$ and $v$ in $\Ll_S$.
\end{proof}

\section{Proof of Lemma~27}\label{app:cw-metathm} 

We now prove~\Cref{lem:meta-thm}, which we first recall for convenience.
\cwmetathm*

\begin{proof}
 Using the algorithm of Fomin and Korhonen~\cite{FominK22}, in time $\Oh(n^2)$ we can compute a clique expression $t$ of constant width that constructs the input graph $G$. We may view $t$ as a rooted tree labeled with a constant-size alphabet: internal nodes are labeled with applied clique expression operations (relabel/join/union), while every leaf is labeled with a pair consisting of its initial label in $t$ and its color in $G$. Then it is straightforward to write an $\mathsf{MSO}$ formula $\alpha(y,z)$ such that $t\models \alpha(u,v)$ if and only if $u$ and $v$ are leaves of $t$ (aka vertices of $G$) and $u$ and $v$ are adjacent in $G$. Let $\varphi'(\tup x)$ be the formula obtained from $\varphi(\tup x)$ by replacing every adjacency test $E(y,z)$ with formula $\alpha(y,z)$, and further stipulating that every $x\in \tup x$ is a leaf of $t$. Then for any $\tup u\in G^{\tup x}$, $G\models \varphi'(\tup u)$ if and only if $t\models \varphi'(\tup u)$. Now, to obtain the enumeration statement it suffices to apply the enumeration algorithm for $\mathsf{MSO}$ queries on labeled trees, due to Kazana and Segoufin~\cite{KazanaS13}, to $t$ and $\varphi'$. Similarly, for the query answering result one can use~\cite[Theorem~6.1.3]{KazanaThesis}. 
\end{proof}

\section{Proof of Theorem~29}\label{app:td-canonization} 

We now prove~\Cref{thm:td-canonization}, which we restate for convenience.

\tdcanonization*

 We will use the following result observed by Bouland et al.~\cite{BoulandDK12}. (In fact, this result was a direct inspiration for Theorem~\ref{thm:progressive}, as it corresponds to the case $r=\infty$ there.)
 
 \begin{lemma}[Lemma~7 of~\cite{BoulandDK12}]\label{lem:roots}
  There exists a computable function $f\colon \N\to \N$ such that in every connected graph $G$ of treedepth $d$, there exists only at most $d$ different vertices $u$ such that the treedepth of $G-u$ is strictly smaller than $d$.
 \end{lemma}

 In fact, as noted in~\cite{BoulandDK12}, Lemma~\ref{lem:roots} also follows from an earlier work of Giannopoulou et al.~\cite{DvorakGT12}, and a more recent work of Chen et al.~\cite{ChenCDFHNPPSWZ21} provides a bound $f(d)\in d^{\Oh(d)}$.
 
 Next, suppose $(\str M,\preceq_{\str M})$ and $(\str N,\preceq_{\str N})$ are two ordered $\Sigma$-structures, that is, $\Sigma$-structures equipped with orders on universes. We can compare $(\str M,\preceq_{\str M})$ and $(\str N,\preceq_{\str N})$ lexicographically as follows:
 \begin{itemize}
  \item If $|\str M|\neq |\str N|$, then $(\str M,\preceq_{\str M})$ is smaller than $(\str N,\preceq_{\str N})$ if and only if $|\str M|<|\str N|$.
  \item Otherwise, for each $R\in \Sigma$, sort the tuples in relation $R$ in $\str M$ and in $\str N$ lexicographically with respect to $\preceq_{\str M}$, respectively $\preceq_{\str N}$. Then replace every element of $\str M$ with its index in $\preceq_{\str M}$, and similarly for $\str N$. Compare the lists obtained in this way for $\str M$ and $\str N$. If for every $R\in \Sigma$ they are equal, then $(\str M,\preceq_{\str M})$ and $(\str N,\preceq_{\str N})$ are isomorphic. Otherwise, the first relation $R\in \Sigma$ for which the lists are different determines whether $(\str M,\preceq_{\str M})$ or $(\str N,\preceq_{\str N})$ is smaller. (We assume here a fixed order on the relation names in $\Sigma$.)
 \end{itemize}
 The following lemma follows by a straightforward implementation of the comparison procedure above. Here, for a structure $\str M$, by $\|\str M\|$ we mean $|\str M|$ plus the total length of all tuples present in relations of $M$, that is, $\|\str M\|=|\str M|+\sum_{R\in \Sigma} \mathrm{arity}(R)\cdot |R^{\str M}|$.
 
 \begin{lemma}\label{lem:sort}
  Let $\Sigma$ be a fixed signature. Then
  given two ordered structures $(\str M,\preceq_{\str M})$ and $(\str N,\preceq_{\str N})$, one can in time $\Oh(n\log n)$ decide whether they are isomorphic or which is lexicographically smaller than the other, where $n=\|\str M\|+\|\str N\|$.
 \end{lemma}
 
 Finally, the following lemma is easy.
 
 \begin{lemma}\label{lem:sparse}
  If a $\Sigma$-structure $\str M$ has treedepth at most $d$, then $\|\str M\|\leq g(\Sigma,d)\cdot |\str M|$ for a computable function $g$.
 \end{lemma}
 \begin{proof}
  Let $F$ be an elimination forest of $\str M$ of depth at most $d$, that is, a rooted forest on the universe of $\str M$ such that for every edge $uv$ of the Gaifman graph of $\str M$, $u$ and $v$ are in the ancestor/descendant relation in $\str M$. It follows that for every tuple $\tup u$ in any relation in $\str M$, the elements featured in $\tup u$ lie on one root-to-leaf path in $F$. Therefore, $\tup u$ can be uniquely described by specifying the deepest element featured in $\tup u$, its position within the tuple $\tup u$, and the depths in $F$ of all the other elements of $\tup u$. Since there are at most $|\str M|\cdot k\cdot d^{k-1}$ descriptions of this form, where $k$ is the arity of $\tup u$, summing such bounds for all relations in $\Sigma$ establishes the result.
 \end{proof}

 With all the tools prepared,
 we can now proceed to the proof of \Cref{thm:td-canonization}.

\begin{proof}[Proof of \Cref{thm:td-canonization}]
 By induction on $d$, we prove the following statement: for every finite signature $\Sigma$ and $d\in \N$, there exists a canonical mapping $\cn_{\Sigma,d}$ for the class of all $\Sigma$-structures of treedepth at most $d$ that can be computed in time $g(\Sigma,d)\cdot n\log^2 n$, for some computable function~$g$. For $d=0$, the class may contain only the empty structure, so the result is trivial.
 
 Let then $\str M$ be a $\Sigma$-structure whose Gaifman graph $G$ has treedepth at most $d$. We may assume that $G$ is connected, for in the case of a disconnected $G$ we may proceed as follows:
 \begin{itemize}
  \item Apply the argument for connected graphs to every connected component $C$ of $G$, yielding an ordering $\cn_{\Sigma,d}(C)$.
  \item Sort the obtained ordered structures $(C,\cn_{\Sigma,d}(C))$ first by the number of elements, and then lexicographically.
 \end{itemize}
 Note that if for $p\leq n$ by $m_p$ we denote the number of components $C$ with $|C|=p$, then sorting such components $C$ using, say, mergesort takes time $$g(\Sigma,d)\cdot m_p\log m_p\cdot p\log p\leq g(\Sigma,d)\cdot m_p p\log^2 n$$ for a computable $g$, due to Lemmas~\ref{lem:sort} and~\ref{lem:sparse}. Summing through all relevant $p$, the procedure takes time $g(\Sigma,d)\cdot n\log^2 n$.

 Once $G$ is assumed to be connected, let us define
 $$X=\{u\in V(G)\colon \mathrm{treedepth}(G-u)<\mathrm{treedepth}(G)\}.$$
 By Lemma~\ref{lem:roots}, we have $|X|\leq f(d)$ for some computable $f$.
 Noting that every class of bounded treedepth can be defined in first-order logic (for instance, using~\Cref{cor:fo-def}), we can write a first-order formula $\varphi(x)$ that selects exactly the vertices of $X$. By applying the enumeration algorithm for $\mathsf{MSO}_2$ queries on structures of bounded treewidth due to Kazana and Segoufin~\cite{KazanaS13}, we can compute $X$ in time $h(\Sigma,d)\cdot n$ for a computable function $h$.
 
 \newcommand{\dm}{\mathrm{dom}\,}
 \newcommand{\partto}{\rightharpoonup}
 
 For every relation $R\in \Sigma$, say of arity $k$, and a partial function $\eta\colon \{1,\ldots,k\}\partto \{1,2,\ldots,|X|\}$, introduce a relation name $R_{\eta}$ of arity $k-|\dm \eta|$, where $\dm \eta$ denotes the domain of $\eta$. Let $\Sigma'$ be a signature consisting of all those relation names.
 Next, for every bijection $\sigma\colon \{1,2,\ldots,|X|\}\to X$, we define a new $\Sigma'$-structure $\str M_\sigma$ as follows:
 \begin{itemize}
  \item The universe of $\str M_\sigma$ is $V(G)\setminus X$.
  \item For every relation $R\in \Sigma$, say of arity $k$, and a partial function $\eta\colon \{1,\ldots,k\}\partto \{1,2,\ldots,|X|\}$, we interpret $R_\eta$ in $\str M_\sigma$ as the set of all tuples $(a_1,\ldots,a_{k'})\in (V(G)\setminus X)^{k'}$ such that $\str M\models R(b_1,\ldots,b_k)$, where $k'=k-|\dm \eta|$, $b_i=\sigma(\eta(i))$ if $i\in \dm \eta$, and otherwise $b_i=a_{\iota(i)}$, where $\iota$ is the order-preserving bijection from $\{1,\ldots,k\}\setminus \dm \eta$ to $\{1,\ldots,k'\}$. 
 \end{itemize}
 Note that this gives $|X|!\leq f(d)!$ $\Sigma'$-structures $\str M_\sigma$, and, thanks to Lemma~\ref{lem:sparse}, each of the can be computed in time $h(\Sigma,d)\cdot n$ for some computable $h$. Moreover, each structure $\str M_\sigma$ has treedepth at most $d-1$.
 
 We may now apply the algorithm provided by the induction to compute, for each permutation $\sigma$ as above, an ordering $\preceq_\sigma=\cn_{\Sigma',d-1}(\str M_\sigma)$. Let $\preceq_\sigma'$ be obtained from $\preceq_\sigma$ by prepending the vertices of $X$ in the order $\sigma$. Note that the obtained set of ordered structures 
 $$\left\{(\str M,\preceq_{\sigma}')\colon \sigma\colon \{1,2,\ldots,|X|\}\to X\textrm{ is a bijection}\right\}$$
 is isomorphism-invariant. Therefore, we can set $\cn_{\Sigma,d}(G)$ to be the order $\preceq_{\sigma'}$ for which the structure $(\str M,\preceq_{\sigma'})$ is lexicographically the smallest among the ones above, where we use the algorithm of Lemma~\ref{lem:sort} to compare the structures.
 
 As for the time complexity, the algorithm presented above invokes at most $f(d)!$ calls to the algorithm computing $\cn_{\Sigma',d-1}$. Therefore, there are at most $d$ levels of recursion, and the total work used on every level is $h(\Sigma,d)\cdot n\log^2 n$ for a computable function $h$. It follows that the whole algorithm runs in time $g(\Sigma,d)\cdot n\log^2 n$ for a computable function $g$.
\end{proof}